\documentclass[letterpaper,11pt]{article}
\usepackage{subcaption}

\usepackage[margin=1in]{geometry}
\usepackage{amsmath, amsthm, amssymb, thmtools, mleftright}
\usepackage{xspace, paralist, enumitem, multirow, booktabs}
\usepackage[dvipsnames]{xcolor}
\usepackage{bm, bbm}
\usepackage[T1]{fontenc}
\usepackage{graphicx}
\usepackage{algorithm}
\usepackage[noend]{algpseudocode}
\usepackage{url}
\usepackage{pbox}
\usepackage{titletoc}
\usepackage{tikz}
\usepackage{framed}
\usepackage{aliascnt}
\usepackage{natbib}
\usepackage{wrapfig}
\usepackage[table]{xcolor} % Use [table] option for colortbl functionality

\usetikzlibrary{positioning, shapes, arrows.meta, calc}

\makeatletter
\def\th@plain{%
  \thm@notefont{}% same as heading font
  \itshape % body font
}
\def\th@definition{%
  \thm@notefont{}% same as heading font
  \normalfont % body font
}
\makeatother
\newtheorem{theorem}{Theorem}[section]
\newtheorem{lemma}[theorem]{Lemma}
\newtheorem{corollary}[theorem]{Corollary}

\newtheorem{claim}[theorem]{Claim}

\newtheorem{definition}[theorem]{Definition}

\newtheorem*{remark}{Remark}

\newlist{thmparts}{enumerate}{1}
\setlist[thmparts]{labelindent=\parindent,leftmargin=*,itemsep=2pt,font=\normalfont,label=(\thetheorem.\arabic*)}
\declaretheoremstyle[%
  spaceabove=-6pt,%
  spacebelow=6pt,%
  headfont=\normalfont\itshape,%
  postheadspace=1em,%
  qed=\qedsymbol%
]{mystyle} 

\RequirePackage{caption, float}
\algrenewcomment[1]{\hfill \textcolor{blue}{$\triangleright$ #1}}
\algnewcommand{\LineComment}[1]{\State \textcolor{blue}{$\triangleright$ #1}}
\algdef{SE}[SUBALG]{Indent}{EndIndent}{}{\algorithmicend\ }%
\algtext*{Indent}
\algtext*{EndIndent}

\newcommand{\wO}{\widetilde{O}}

\newcommand{\EE}{\mathbb{E}}

\newcommand{\cA}{\mathcal{A}}

\newcommand{\cD}{\mathcal{D}}

\newcommand{\eps}{\varepsilon}

\makeatletter
\newcommand{\tpmod}[1]{{\@displayfalse\pmod{#1}}}
\makeatother

\makeatletter
\newcommand\squeezepar{\@startsection{paragraph}{4}{\z@}{1.5ex \@plus1ex \@minus.2ex}{-1em}{\normalfont\normalsize\bfseries}}
\makeatother
\usepackage{hyperref}
\hypersetup{
  colorlinks,
  citecolor=blue,
  linkcolor=magenta,
  urlcolor=blue}

\newaliascnt{algocf}{algorithm}
\aliascntresetthe{algocf}

\begin{document}

\title{Efficient Algorithms for Adversarially Robust Approximate Nearest Neighbor Search}

\date{\today}

%%%%%%%%%%%%%%%%%%%%%%%%%%%%%%%%%%%%%%%%%%%%%%%%%%%%%%%%%
%%%%%%%%%%%%%%%%%%% UNCOMMENT FOR PUBLISHING %%%%%%%%%%%%%%%%%%%
%%%%%%%%%%%%%%%%%%%%%%%%%%%%%%%%%%%%%%%%%%%%%%%%%%%%%%%%%

\author{%
  Alexandr Andoni\footnotemark[1]
  \and
  Themistoklis Haris\footnotemark[2]
  \and
  Esty Kelman\footnotemark[3]
  \and
  Krzysztof Onak\footnotemark[4]
}

\footnotetext[1]{Department of Computer Science, Columbia University. Email: \texttt{andoni@cs.columbia.edu}}
\footnotetext[2]{Computer Science Department, Boston University. Email: \texttt{tharis@bu.edu}}
\footnotetext[3]{Department of Computer Science and Department of Computing \& Data Sciences, Boston University and CSAIL, Massachusetts Institute of Technology.  Supported in part by the National Science Foundation under Grant No. 2022446 and in part by NSF TRIPODS program (award DMS-2022448). Email: \texttt{ekelman@mit.edu}.}
\footnotetext[4]{Faculty of Computing and Data Sciences, Boston University. Email: \texttt{konak@bu.edu}}

%%%%%%%%%%%%%%%%%%%%%%%%%%%%%%%%%%%%%%%%%%%%%%%%%%%%%%%%%%%%%%%%%%%%%%%%%%%%
%%%%%%%%%%%%%%%%%%%%%%%%%%%%%%%%%%%%%%%%%%%%%%%%%%%%%%%%%%%%%%%%%%%%%%%%%%%%

\maketitle

\thispagestyle{empty}

\begin{abstract}
We study the Approximate Nearest Neighbor (ANN) problem under a powerful adaptive adversary that controls both the dataset and a sequence of $Q$ queries.

Primarily, for the high-dimensional regime of $d = \omega(\sqrt{Q})$, we introduce a sequence of algorithms with progressively stronger guarantees. We first establish a novel connection between adaptive security and \textit{fairness}, leveraging fair ANN search~\citep{aumuller2022sampling} to hide internal randomness from the adversary with information-theoretic guarantees. To achieve data-independent performance, we then reduce the search problem to a robust decision primitive, solved using a differentially private mechanism~\citep{hassidim2022adversarially} on a Locality-Sensitive Hashing (LSH) data structure. This approach, however, faces an inherent $\sqrt{n}$ query time barrier. To break the barrier, we propose a novel concentric-annuli LSH construction that synthesizes these fairness and differential privacy techniques. The analysis introduces a new method for robustly releasing timing information from the underlying algorithm instances and, as a corollary, also improves existing results for fair ANN.

In addition, for the low-dimensional regime $d = O(\sqrt{Q})$, we propose specialized algorithms that provide a strong ``for-all'' guarantee: correctness on \textit{every} possible query with high probability. We introduce novel metric covering constructions that simplify and improve prior approaches for ANN in Hamming and $\ell_p$ spaces.
\end{abstract}

\newpage

\tableofcontents

\newpage

\section{Introduction}
Randomness is a crucial tool in algorithm design, enabling resource-efficient solutions by circumventing the worst-case scenarios that plague deterministic approaches \citep{motwani1996randomized}. The classical analysis of such algorithms assumes an \textbf{oblivious} setting, where data updates and queries are fixed beforehand. However, this assumption breaks down in the face of an \textbf{adaptive adversary}, who can issue queries based on the algorithm's previous outputs. These outputs can leak information about the algorithm's internal randomness, allowing an adversary to construct query sequences that maliciously break the algorithm's performance guarantees \citep{hardt2013robust, gribelyuk2024strong}.

Significant progress has been made in designing adversarially robust algorithms for \textbf{estimation problems}, where the output is a single value \citep{lai2020adversarial, hassidim2022adversarially, chakrabarti2021adversarially, attias2024framework, ben2022adversarially, woodruff2022tight, cherapanamjeri2023robust}. A common defense involves sanitizing the output, for example, by rounding or adding noise, often borrowing techniques from differential privacy to ensure the output reveals little about the algorithm's internal state \citep{hassidim2022adversarially, attias2024framework, beimel2022dynamic}. However, these techniques do not directly apply to \textbf{search problems}. In a search problem, the algorithm must return a specific element from a given dataset. Outputting a raw data point can leak substantial information, and there is no obvious way to add noise or otherwise obscure the output without violating the problem's core constraint of returning a valid dataset element.

Perhaps the most fundamental search problem is \textit{Approximate Nearest Neighbor (ANN) Search}, which has numerous applications ranging from data compression and robotics to DNA sequencing, anomaly detection 
and Retrieval-Augmented Generation (RAG) \citep{santalucia1996improved, kalantidis2014locally, ichnowski2015fast, verstrepen2014unifying, tagami2017annexml, bergman2020deep,han2023hyperattention,  kitaev2020reformer}. Formally, we define the problem as follows:
\begin{definition}[The $(c,r)$-ANN problem]
Given a dataset $S$ of $n$ points in a metric space $(\mathcal{M},||\cdot||)$ and a radius $r > 0$, let $B_S(q,r) := \{p \in S:||p-q|| \leq r\}$. Given a query point $q \in \mathcal{M}$ and approximation parameter $c\geq 1$, the goal is to build a data structure which finds a point in $B_S(q,cr)$ if $B_S(q,r) \neq \emptyset$. If $B_S(q,cr) = \emptyset$, the algorithm is required to answer $\perp$. Apart from queries, the dataset $S$ itself may also be \textit{obliviously updated} via additions or deletions of points. 

A data structure for solving this problem is evaluated in terms of its space complexity, query runtime, and update runtime as functions of the dataset size $n$, the parameter $c$ and possibly the total number of queries.
\end{definition}

Achieving the desired trade-off of sublinear query time and near-linear space has largely been possible only through randomization. Indeed, one of the most prominent family of algorithms for ANN is based on \textit{Locality-Sensitive Hashing (LSH)}, which has been the subject of a long and fruitful line of research in the oblivious setting \citep{gionis1999similarity, jafari2021survey, andoni2009nearest, andoni18-icm, andoni2017lsh, andoni2016lower, andoni2017nearest, andoni2017optimal, indyk1998approximate, broder1998min, andoni2022learning}. ANN Algorithms that rely on LSH achieve query time complexity of $\widetilde{O}(dn^{\rho})$\footnote{We use the $\widetilde{O}$ notation to hide polylogarithmic factors.} and space complexity $\widetilde{O}(n^{1+\rho})$, where $d$ is the dimension of $\mathcal{M}$ and $\rho = \rho(c) \in (0,1)$ is a fixed constant depending on $c$ and the LSH construction\footnote{For example, when $\mathcal{M} = \{0,1\}^k$ under the Hamming distance and $c \geq 1$ is the approximation parameter, the state-of-the-art construction of \citet{andoni2015optimal} yields $\rho = \frac{1}{2c-1}$. We shall use $\rho$ and $\rho(c)$ interchangeably.}.

The vulnerability of these classical randomized structures was recently highlighted by \citet{kapralov2024adversarial}, who demonstrated an attack on standard LSH data structures. They showed that an adaptive adversary can use a polylogarithmic number of queries to learn enough about the internal hash functions to force the algorithm to fail. Inspired by their work, which relies on certain structural properties of the dataset (e.g., an ``isolated'' point), we consider a powerful adversarial model where the \textit{adversary chooses both the dataset and the sequence of queries}. We study the following question:
\begin{center}
    \textit{Can search problems like ANN be solved efficiently in the face of adversarial queries?}
\end{center}

Our adversarial model for the $(c,r)$-ANN problem empowers the adversary by giving them unbounded computational resources, the ability to specify the dataset $S$ completely, the ability to obliviously choose a sequence of updates to that dataset, and, perhaps most importantly, the ability to choose each query based on the full history of their interaction with the algorithm. While alternative adversarial models could be considered, our chosen framework maximizes generality to demonstrate that efficient data structures remain achievable even under broad constraints.

\subsection{Roadmap and Discussion}
We begin by placing our results in the broader context of adversarial robustness, which has primarily followed two paradigms. The first is \emph{robustification via privacy} \citep{hassidim2022adversarially}, where multiple independent copies of an algorithm are aggregated using a differentially private mechanism. Privacy obfuscates the internal randomness, preventing an adaptive adversary from predicting future behavior. The second is \emph{for-all} algorithms \citep{gilbert2007one}, which succeed on every input, but typically incur high runtime or space overhead due to discretization and union bounds, making them problem-specific.  

The ANN problem poses a further challenge: it is a \emph{search} problem rather than an estimation problem, so adding noise to outputs would destroy correctness. One natural way to resolve this is by reducing search to collections of \emph{decision} problems that admit privacy-based robustification. In ANN, partitioning the dataset and solving weak decision problems in each part allows robust recovery of a valid neighbor while retaining sublinear query time and space. We believe that the idea of reducing from search to robust decision may extend to other search problems.  

Our main conceptual contribution is identifying a third route to robustness: \emph{fairness}. Informally, fair algorithms avoid bias toward particular outputs and maintain statistical independence across repeated queries. Such algorithms are inherently robust, as they reveal little information about internal randomness. We formalize this intuition, showing that fairness can be viewed as a valuable tool towards establishing robustness. From this perspective, fairness and differential privacy can be seen as two instantiations of a common underlying principle:  
\begin{center}
    \textbf{Robustness follows from stability with respect to internal randomness.}
\end{center}

Finally, these approaches can be combined. By merging fairness-based robustness with privacy-based robustification and exploiting ANN’s geometric structure via a concentric annuli construction, we achieve improved tradeoffs for robust nearest neighbor search.

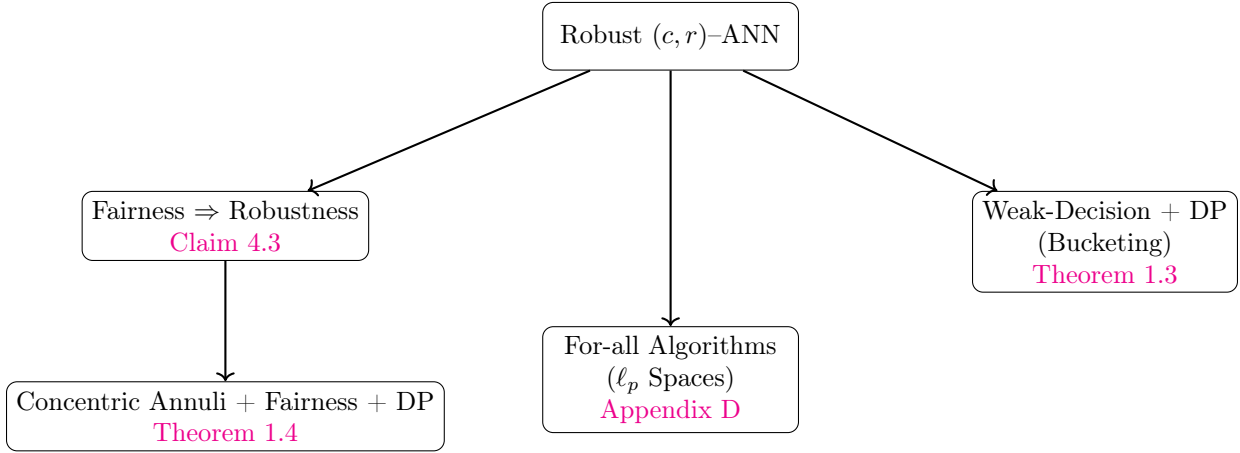
\begin{figure}[h]
\centering
\begin{tikzpicture}[
    node distance=1.6cm and 2.6cm,
    every node/.style={
        draw,
        rounded corners,
        align=center,
        minimum width=3.4cm,
        minimum height=0.9cm,
        font=\small,
        scale=1
    },
    arrow/.style={->, thick}
]

% ================= Root =================
\node (root) {Robust $(c,r)$--ANN};

% ================= Left branch: Fairness =================
\node (fairclaim) [below left=of root, xshift=0.3cm] 
    {Fairness $\Rightarrow$ Robustness\\ \autoref{claim:fairness-robustness}};
\node (concthm) [below=of fairclaim] 
    {Concentric Annuli + Fairness + DP\\ \autoref{thm:concentric-ann-robust-best}};

\draw[arrow] (root) -- (fairclaim);
\draw[arrow] (fairclaim) -- (concthm);

% ================= Right branch: Weak-decision + DP =================
\node (buckthm) [below right=of root, xshift=-0.3cm] 
    {Weak-Decision + DP\\(Bucketing)\\ \autoref{thm:robust-ann-bucketing}};

\draw[arrow] (root) -- (buckthm);

% ================= For-all branch =================
\node (forallthm) [below=3.4cm of root] 
    {For-all Algorithms \\($\ell_p$ Spaces)\\ \autoref{sec:for-all-main}};

\draw[arrow] (root) -- (forallthm);

\end{tikzpicture}
\caption{Roadmap of results.}
\end{figure}

\subsection{Our Results and Techniques}
In this section we present our technical insights in detail
\paragraph{Robustness and Fairness} 
We first recognize a connection between robustness and fairness. \textit{Fair} ANN algorithms output a point uniformly at random from a set of valid near neighbor candidates. Such algorithms have already been rigorously studied in the context of LSH by \citet{aumuller2022sampling}, who also studied notions of approximate fairness. We show that the robust ANN problem can be solved simply by invoking an algorithm for the exact fair ANN problem.
\begin{theorem}
\label{thm:adv-robustness-ann-fair}
Let $n(q,r) := |B_S(q,r)|$ be the $S$-density of the $r$-ball centered at $q \in \mathcal{M}$. There exists an adversarially robust $(c,r)$-ANN algorithm that uses $O(n^{1+\rho(c)}\log(nQ))$ bits of space and $O(d\cdot(n^{\rho(c)}+\frac{n(q,cr)}{n(q,r)+1})\log(nQ))$ time per query.
\end{theorem}

Though very space efficient, the query complexity of this algorithm unfortunately depends on the density ratio $D$ of points between the $cr$-ball and the $r$-ball for a query $q$. The adversary can craft a dataset where this ratio is large, severely degrading performance. This drawback is also shared by the algorithm of \cite{fengdifferential}, though they exhibit a dependency on the density $s:=n(q,cr)$, which is strictly greater than $D$ (see \autoref{table:comparison}).

\begin{remark}
    The link between fairness and robustness is not limited to ANN. From this perspective, fairness is not just a ``nice to have'' property, but is inextricably linked with security. 
\end{remark}

\paragraph{Assumption-Free Searching via Bucketing} 
To mitigate data dependencies, we propose a meta-algorithm that reduces a search problem to a weak decision problem. In this problem, positive instances correspond to the existence of $r$-close neighbors to a query $q$, while negative instances showcase the absence of $cr$-close neighbors. Such a weak decision problem can be solved obliviously simply by using a classic LSH data structure $\mathcal{D}$. Unlike the search problem, an oblivious decider can be robustified by applying the well-known Differential Privacy (DP) obfuscation technique of \citep{hassidim2022adversarially}: we maintain $\sqrt{Q}$ copies of $\mathcal{D}$ and combine their responses in a private manner with respect to the random bits of each copy. 

To perform the search, we then partition $S \in \mathcal{M}^n$ into buckets of size roughly $\sqrt{n}$ and instantiate a copy of the robust weak decider in each bucket. We can use these copies to identify a bucket that contains a suitable point to output and then exhaustively search that bucket to produce the final answer:
\begin{theorem}
\label{thm:robust-ann-bucketing}
There exists an adversarially robust algorithm for the $(c,r)$-ANN problem, successfully answering up to $Q$ queries with probability at least $1-\Theta(\delta)$. The algorithm uses $\widetilde{O}(n^{1 + \rho/(2-\rho)}\sqrt{Q})$ space and $\widetilde{O}(dn^{1/(2-\rho)})$ time per-query, where $\rho = \rho(c) \in (0,1)$.
\end{theorem}
\begin{table}[h!]
\centering
\renewcommand{\arraystretch}{1.5}
\label{tab:comparison}
\begin{tabular}{|l|c|c|c|}
\hline
\textbf{Metric} & Query Time & Space & Update Time \\
\hline
\autoref{thm:adv-robustness-ann-fair} (Fairness) & $\widetilde{O}(d\cdot({\color{black} D} +n^{\rho}))$ & \cellcolor{white}$\widetilde{O}(n^{1+\rho})$ & $\widetilde{O}(d\cdot({\color{black} D} +n^{\rho}))$\\
\hline
\autoref{thm:robust-ann-bucketing} (Bucketing) & $\widetilde{O}(d n^{{\color{black}\frac{1}{2-\rho}}})$ & $\widetilde{O}(\sqrt{Q} \cdot n^{{\color{black}\frac{2}{2-\rho}}})$ & $\widetilde{O}(dn^{\frac{\rho}{2-\rho}}\sqrt{Q})$ \\
\hline
\pbox{6cm}{\vspace{2mm}\autoref{thm:concentric-ann-robust-best} (Concentric Annuli)\\[1mm]
\centering{$\beta = \Theta(\frac{\log\log c}{\log c})$}\vspace{2mm}}
& \cellcolor{white}$\widetilde{O}(dn^{{\color{black}\beta}})$ & $O(\sqrt{Q} \cdot n^{{\color{black}1+\beta}})$ & $\widetilde{O}(dn^\beta\sqrt{Q})$ \\
\hline
\citep{fengdifferential}\footnotemark & $O(d\cdot{\color{BrickRed} s}\cdot n^{\rho})$ & $O(\sqrt{Q} \cdot {\color{BrickRed}s} \cdot n^{1+\rho})$ & $\widetilde{O}(dn^\rho\cdot {\color{BrickRed} s}\cdot \sqrt{Q})$\\
\hline
\end{tabular}
\caption{Algorithms for the $(c,r)$-ANN problem in LSH-equipped metric spaces.}
\label{table:comparison}
\vspace{-2mm}
\end{table}
\footnotetext{The work of \cite{fengdifferential} concurrently studies the robust ANN problem. We present a comparison of our results with their algorithm, as well as a more extended discussion of related work, in \autoref{sec:related-work-more}.}

\paragraph{Breaking the $\sqrt{n}$ Barrier via Concentric LSH Annuli}
The bucketing method yields a query time complexity that is always at least $O(\sqrt{n})$, which is not ideal considering that LSH methods can induce the exponent of $n$ to be arbitrarily close to $0$ by increasing the approximation parameter $c$. To go beyond this barrier, we introduce a \textit{concentric annuli construction}.

We partition the $(r, cr)$-annulus into several smaller, concentric sub-annuli and apply the fair ANN algorithm $\mathcal{A}_{\text{fair}}$ within each one. A simple counting argument guarantees that at least one of these sub-annuli must have a low point-density ratio, implying that $\mathcal{A}_{\text{fair}}$ within that annulus terminates within an acceptable runtime threshold with probability at least $0.9$. For each annulus that does not exceed this threshold, we obtain an estimate to the probability that the corresponding $\mathcal{A}_{\text{fair}}$ copy terminates quickly. We then pick a favorable annulus to run a held-out \textit{testing} copy of $\mathcal{A}_{\text{fair}}$ and output a point. To maintain robustness however, we must be careful to ensure that annulus selection is done in a way such that the adaptive adversary cannot overfit to our internal randomness. To do this we apply the DP robustification framework on the selection process of the annuli by estimating the probabilities each copy of $\mathcal{A}_{\text{fair}}$ terminates within a predetermined runtime threshold. Our algorithm is both assumption-free and enjoys a better runtime than $O(\sqrt{n})$.

\begin{theorem}
    \label{thm:concentric-ann-robust-best}
    There exists a robust algorithm for solving the $(c,r)$-ANN problem that that uses space $\widetilde{O}(\sqrt{Q}\cdot n^{1+\beta})$, where $\beta=\min_{k\in\mathbb{Z}_{\geq 1}}\max\{\rho(c^{1/k}), 1/k\}$. The query runtime is $\widetilde{O}(dn^{\beta})$ time with probability at least $1-\delta$.
\end{theorem}
% \aanote{how is $c'$ used here? it is only defined.}

For many metric spaces, the value of $\beta$ resolves nicely. For the hypercube $\{0,1\}^d$ under the Hamming distance we have $\rho(c) = \frac{1}{2c-1}$, which yields $\beta = \Theta(\frac{\log\log c}{\log c}) \to 0$ as $c\to\infty$, which is not the case with the exponent $\frac{1}{2-\rho(c)}$ of \autoref{thm:robust-ann-bucketing}. 
% Though a more concise closed formula for $\beta$ is hard to obtain,\aanote{can't we say that $\beta=\Theta(\log c/\log\log c)$} its value goes to zero as $c\to\infty$, which is not the case with the exponent $\frac{1}{2-\rho}$ of our previous approach. For example, for solving the ANN problem in $\ell_2$ with $c=10$ we get $\beta = 1/4$, a strongly sublinear result. 
As a corollary, this technique also allows us to achieve purely sublinear time for a class of ``relaxed'' fair ANN problems.

\paragraph{For-all Algorithms}
For low-dimensional metric spaces, we develop algorithms for ANN that provide a powerful \textit{for-all guarantee}: with high probability, the data structure correctly answers \textit{every possible} query $q \in \mathcal{M}$. Our approach builds on a discretization technique applied to an LSH data structure, a paradigm explored in prior work \citep{cherapanamjeri2020adaptive, cherapanamjeri2024terminal, bateni2024efficient}. We refine this line of research by introducing a novel, simpler metric covering construction, improving the space complexity by a logarithmic factor, and using sampling to improve the time complexity by a factor of $d$. We present our result for the Hamming hypercube below, including results for $\ell_p$ spaces in \autoref{sec:for-all-main}.

\begin{theorem}
For the $(c,r)$-ANN problem in the $d$-dimensional Hamming hypercube $\{0,1\}^d$, there exists an algorithm that correctly answers all possible queries with at least $0.99$ probability. The space complexity is $\widetilde{O}(d\cdot n^{1+\rho+o(1)})$ and query time is $\widetilde{O}(d \cdot n^\rho)$, where $\rho = \frac{1}{2c-1}$.
\end{theorem}

\begin{remark}[The Price of For-All Algorithms]
Despite their remarkable guarantees, \textit{for-all} algorithms have significant drawbacks. Their space complexity scales by a factor of $d$, making them intractable for high-dimensional metric spaces. This is a direct consequence of the large number of hash functions required to ensure a tiny probability of error for any query. Furthermore, these algorithms lack the generality of their adaptive counterparts; they are metric-space dependent and must be tailored to the specific metric space being used.
\end{remark}

\section{Related and Concurrent Work}
\label{sec:related-work-more}
The challenge of designing algorithms robust to adversarial queries is well-studied, particularly in privacy and statistics \citep{bassily2015more, smith2017information, bassily2016algorithmic}, where Differential Privacy is a central tool for ensuring robustness \citep{dwork2015generalization, dinur2023differential}. The question of adversarial robustness was formally introduced to streaming algorithms by Ben-Eliezer et al.~\citep{ben2022framework}, motivated by attacks on linear sketches \citep{hardt2013robust}, and has since inspired a long line of work on robustifying various streaming algorithms \citep{hassidim2022adversarially, chakrabarti2021adversarially, lai2020adversarial, chakrabarti2024finding, stoeckl2023streaming, woodruff2022tight, ben2022adversarially}.

Our work is most directly inspired by the framework of Hassidim et al.~\citep{hassidim2022adversarially}, who used Differential Privacy to solve estimation problems robustly, and by Cherapanamjeri et al.~\citep{cherapanamjeri2023robust}, who applied this framework with low query time overhead. While we adapt a similar approach, their methods are fundamentally limited to estimation and don't extend to search problems like NNS, where the output must be a specific dataset element. The difficulty of robust search is further highlighted by Beimel et al.~\citep{beimel2022dynamic}, who established lower bounds showing that robust algorithms for certain search problems are inherently slower than their oblivious counterparts, motivating our investigation.

Different works further reinforce the unique challenges of robust search. Work on robust graph coloring, for example, also requires techniques beyond simple noise addition due to its discrete output space \citep{chakrabarti2021adversarially, BehnezhadRW25}. Our approach is also distinct from Las Vegas LSH constructions \citep{pham2016scalability, wei2022optimal}. While these methods guarantee no false negatives, they remain vulnerable to adversaries who can inflate their expected runtime \citep{kapralov2024adversarial}. Our focus, in contrast, is on robustifying traditional Monte Carlo algorithms.

Finally, our approach builds on the use of discretization and net-based arguments to achieve 'for-all' guarantees for ANN. This technique was previously used for robust distance estimation \citep{cherapanamjeri2020adaptive}, for ANN in conjunction with partition trees \citep{cherapanamjeri2024terminal} and for efficient centroid-linkage clustering \citep{bateni2024efficient}. We contribute a simpler and more streamlined construction that offers a modest performance improvement over this prior work.

\subsection{Comparison with \texorpdfstring{\citet{fengdifferential}}{Feng et al (2025)}}
Our work was developed concurrently and independently with \citet{fengdifferential}. Our approaches, assumptions, and performance guarantees differ significantly.

\begin{description}
    \item[Methodology] \citet{fengdifferential} propose a method tightly coupled to the structure of DP noise via a reduction to the private selection problem. In contrast, our ``search-to-decision'' and fairness frameworks are more general, treating the DP component as a black-box primitive.

    \item[Assumptions] Their algorithm's complexity depends on a near-neighbor density bound $s$, where $|B_S(q,cr)| \leq s$. We present the first algorithms whose query runtimes are independent of the input dataset, making them robust to worst-case data distributions.
    
    \item[Performance] Their query time scales \textit{multiplicatively} with the number of points in the annulus, $|B_S(q,cr)|$, while our algorithms are either purely sublinear or their query time depends \textit{additively} only on the density ratio $D = \frac{|B_S(q,cr)|}{|B_S(q,r)|}$. Crucially, this dependency on $D$ does not affect our space complexity, which still grows by an additional factor of $\sqrt{Q}$.
\end{description}

\section{Preliminaries}
In this section we present necessary definitions of the problems we address and the computational models we consider. For a comprehensive notation table, see \autoref{table:notation}. For background on differential privacy theorems and concepts, see \autoref{appx:dp-review}.
\paragraph{The Adversarial Robustness Model} An algorithm is adversarially robust if it correctly answers a sequence of adaptively chosen queries with high probability. This is formalized \citep{ben2022framework} through the following interactive game:
\begin{definition}
\label{def:adversarial-robustness}
Consider the following game $\mathcal{G}$ between \textbf{Algorithm} ($\mathcal{A}$) and \textbf{Adversary} ($\mathcal{B}$):
\begin{enumerate}
    \item \textbf{Setup Phase:} The adversary chooses a dataset $S$. The algorithm $\mathcal{A}$ then uses its private internal randomness $R_{\text{setup}} \in \{0,1\}^*$ to preprocess $S$ and build a data structure $D$. The adversary may know the code for $\mathcal{A}$ but not the specific instance of $R_{\text{setup}}$.

    \item \textbf{Query Phase:} The game proceeds for $Q$ rounds. In each round $i \in [Q]$:
    \begin{itemize}
        \item The adversary adaptively chooses a query $q_i$. This choice can depend on the dataset $S$ and the history of all previous queries and their corresponding answers, $(q_1, a_1), \ldots, (q_{i-1}, a_{i-1})$.
        \item The algorithm $\mathcal{A}$ uses its data structure $D$ and potentially new private randomness $R_i \in \{0,1\}^*$ to compute and return an answer $a_i$.
    \end{itemize}

    \item \textbf{Winning Condition:} The algorithm fails if there exists at least one round $i \in [Q]$ for which the answer $a_i$ is an incorrect response to the query $q_i$.
\end{enumerate}

We say that an algorithm $\mathcal{A}$ is $\delta$-adversarially robust if for any dataset and any strategy the adversary can employ, the probability that the algorithm fails is at most $\delta$. The probability is taken over the algorithm's entire internal randomness ($R_{\text{setup}}, R_1, \ldots, R_Q$).
\end{definition}

\begin{remark}[Oblivious Updates]
    Our framework also supports oblivious updates to the dataset. The adversary selects a series of update timesteps in advance in the form of additions or deletions of points. These updates are interleaved with the query phases during the game. We briefly discuss how efficiently updates can be implemented in our algorithms; we are not concerned with robustness as our algorithms works with respect to arbitrarily chosen datasets $S$.
\end{remark}
\paragraph{Approximate Nearest Neighbor Search and LSH}
% We outline the formal problem setup for nearest neighbor search:
% \begin{definition}[Metric Balls]
% Consider a metric space $\mathcal{M} = (M,||\cdot||)$, and let $S \subset \mathcal{M}$. We define a \textbf{metric ball} $B_S(x,r)$ on $S$ of radius $r$ centered at $x \in \mathcal{M}$ as:
% \begin{align*}
%     B_S(x,r) := \{p \in S\mid ||x-p|| \leq r\},\quad\text{with }\,\, n(x,r) := |B_S(x,r)|
% \end{align*}
% \end{definition}
In the \textit{Nearest Neighbors} problem, we seek to find a point in our input dataset that minimizes the distance to some query point.

\begin{definition}[ANN]
Let $c > 1$ and $r > 0$ be positive constants. In the $(c,r)$--\textbf{Approximate Nearest-Neighbors Problem (ANN)} we are given as input a set $S \subset \mathcal{M}$ with $|S| = n$ and a sequence of queries $\{q_i\}_{i=1}^Q$ with $q_i \in \mathcal{M}$. For each query $q_i$, if there exists $p \in B_S(q_i,r)$, we are required to output some point $p' \in B_S(q_i,cr)$. If $B_S(q_i,cr) = \emptyset$, we are required to output $\perp$. In the case where $B_S(q_i,r)=\emptyset \neq B_S(q_i,cr)$ we can either output a point from $B_S(q_i,cr)$ or $\perp$. Our algorithm should successfully satisfy these requirements with probability at least $2/3$.
\end{definition}
A prevalent method for solving ANN is Locality Sensitive Hashing (LSH):
\begin{definition}[Locality Sensitive Hashing, \citep{HIM12}]
A hash family $\mathcal{H}$ of functions mapping $\mathcal{M}$ to a set of buckets is called a $(c,r,p_1,p_2)$--\textbf{Locality Sensitive Hash Family (LSH)} if the following two conditions are satisfied:
\begin{itemize}
    \item If $x,y \in \mathcal{M}$ have $||x-y|| \leq r$, then $\Pr_{h\in \mathcal{H}}[h(x)=h(y)] \geq p_1$.
    \item If $x,y \in \mathcal{M}$ have $||x-y|| \geq cr$, then $\Pr_{h \in \mathcal{H}}[h(x) = h(y)] \leq p_2$.
\end{itemize}
where $p_1 \gg p_2$ are parameters in $(0,1)$. We often assume that computing $h$ in a $d$--dimensional metric space requires $O(d)$ time. We assume that the LSH constructions we consider are \textbf{monotone}, which means that $\Pr[h(x)=h(y)]$ monotonically decreases as $||x-y||$ increases.
\end{definition}
% For instance, in the boolean hypercube $M = \{0,1\}^d$ with $||x-y||$ being the number of positions $j \in [d]$ for which $x_j \neq y_j$, there is a simple monotone LSH family:
% \begin{lemma}
% Consider the family $\mathcal{H} = \{h_j\}_{j=1}^d$ consisting of functions that map $x \in \{0,1\}^d$ to its $j$-th bit $x[j]$. Then, $\mathcal{H}$ is a $(c,r,p_1,p_2)$--LSH where $p_1 = 1 - \frac{r}{d},\text{ and } p_2 = 1 - \frac{cr}{d}$.
% \vspace{-2mm}
% \end{lemma}
% \begin{proof}$\Pr_{h \in \mathcal{H}}[h(x) = h(y)] = \Pr_{j \in [d]}[x[j] = y[j]] = \frac{d-||x-y||}{d} = 1 - \frac{||x-y||}{d}$.
% \end{proof}
% \vspace{-2mm}
Given a construction of a $(c,r,p_1,p_2)$--LSH for a metric space, we can solve the $(c,r)$--ANN problem by amplifying the LSH guarantees. This is done via an \textit{``OR of ANDs''} construction: we sample $L := p_1^{-1} = n^\rho$ hash functions $h_1,...,h_L$ for $\rho(c) \in (0,1)$ by concatenating the outputs of $k= \lceil \log_{1/p_2} n\rceil$ ``prototypical'' LSH functions in $\mathcal{H}$, as shown in \citep{indyk1998approximate}.
\begin{theorem}
\label{thm:lsh_guarantees}
If a $d$--dimensional metric space admits a $(c,r,p_1,p_2)$--LSH family, then we can solve the $(c,r)$--ANN problem on it using $O(n^{1+\rho})$ space and $O(d n^\rho)$ time per query, where $\rho = \frac{\log (1/p_1)}{\log (1/p_2)}$.
\end{theorem}

% As a byproduct of this construction, it is a standard observation that on any query and with high probability, no single bucket contains many outlier points outside of $B_S(q,cr)$ \citep{HIM12, har2019near}.

% \begin{lemma}
%     \label{lem:buckets_not_many_outliers}
%     Let $\mathcal{D} = (h_1,...,h_L)$ be an LSH data structure consisting of $L = O(\log n\cdot n^\rho)$ hash functions that map the metric space $\mathcal{M}$ to $\{0,1\}^k$, where $k=\lceil \log_{1/p_2} n\rceil$. Consider some query $q \in \mathcal{M}$ and let $S_{h_i(q)} := \{p\in S\mid h_i(q)=h_i(p)\}$ be the set of points in $S$ which are hashed in the same bucket as $q$ under $h_i$, for $i \in [L]$. Then, with high probability we have that:
%     \begin{align*}
%         |S_{h_i(q)}| \leq 3\cdot |B_S(q,cr)|
%     \end{align*}
%     Furthermore, if $p\in B_S(q,r)$, then it is contained in at least one bucket with high probability. In other words, it is true that for some $i \in [L]$, with high probability, we have $p \in S_{h_i(q)}$. 
% \end{lemma}

\section{Fairness Implies Robustness}
We first establish a connection between robustness and fairness. We refine the definition of fairness in ANN given by \citet{aumuller2022sampling} to enable the proofs that follow.
\begin{definition}[Exact Fair $(c,r)$-ANN]
A data structure solves the \textbf{Exact Fair $(c,r)$ ANN problem} if, conditioned on all answers
returned so far being correct, the following holds for every round $i$ and every
transcript $T_{i-1}$:
\begin{enumerate}
    \item If $B_S(q_i,r)\neq\emptyset$, then the conditional distribution $\mathcal{L}$ is uniform:
    \[
    \mathcal{L}(a_i \mid T_{i-1}, q_i)
    = \textup{Unif}(B_S(q_i,r)).
    \]
    Otherwise the algorithm outputs $\perp$.
    \item The conditional distribution $\mathcal{L}(a_i \mid T_{i-1}, q_i)$ does not
    depend on the setup randomness $R_{\text{setup}}$.
\end{enumerate}
\end{definition}

We first show that the exact fair ANN algorithm of \citet{aumuller2022sampling} fits our definition. They use a standard $(r,cr,p_1,p_2)$-LSH family with
$\rho=\log(1/p_1)/\log(1/p_2)$.
The data structure consists of
$L=\Theta(n^\rho\log(nQ))$ independent hash tables, yielding
$\widetilde{O}(n^{1+\rho})$ space. Given a query $q$, the algorithm collects the candidate set $C(q)$ from all
tables and applies the exact neighborhood sampling procedure of
\citet[Section~3.5]{aumuller2022sampling}, which uses rejection sampling and fresh per-query randomness.

\begin{theorem}
\label{thm:fair_lsh_prior}
There exists a randomized $(c,r)$-ANN data structure using $O(n^{1+\rho}\log(nQ))$ space with the following properties.
For any fixed query $q$ independent of $R_\text{setup}$, with probability at least $1-\frac{1}{nQ}$ over $R_{\text{setup}}$, the candidate set $C(q)$ contains all points of $B_S(q,r)$. Conditioned on this event, the algorithm returns an element of $B_S(q,r)$ that is exactly uniformly distributed and whose distribution is independent of $R_\text{setup}$. Moreover, the expected query time for $q$ is bounded:
\[
O\!\left(d\cdot\Big(n^\rho + \frac{n(q,cr)}{n(q,r)+1}\Big)\cdot\log(nQ)\right).
\]
\end{theorem}

\begin{proof}
Fix any query $q$ that is independent of $R_{\textup{setup}}$.
For any $p\in B_S(q,r)$, the probability that $p$ fails to collide with $q$
in all $L$ tables is at most $(1-p_1)^L$.
By a union bound over the at most $n$ data points, our choice of $L$ yields
\[
\Pr\limits_{R_\text{setup}}[\exists p\in B_S(q,r)\notin C(q)] \le n(1-p_1)^L \le \frac{1}{nQ}.
\]
Let $\mathcal{E}_q$ denote the complementary event that
$B_S(q,r)\subseteq C(q)$ and condition on $\mathcal{E}_q$.

Conditioned on $\mathcal{E}_q$, the candidate set $C(q)$ contains the entire
$r$-neighborhood of $q$.
The query algorithm then applies the exact neighborhood sampling procedure
of \citet[Section~3.5]{aumuller2022sampling} to the candidate set $C(q)$.
By the correctness of that procedure, the returned point is distributed
\emph{exactly uniformly} over $B_S(q,r)$.

Moreover, the random choices that determine the output (rejection sampling
coins and random swaps) are generated freshly at query time and are independent
of both the preprocessing randomness $R_{\textup{setup}}$ and any prior
interaction transcript.
Since the candidate set $C(q)$ itself depends only on $q$ and
$R_{\textup{setup}}$, it follows that conditioned on $\mathcal{E}_q$ and
for any transcript $T$ the conditional distribution $\mathcal{L}(a \mid T, q)$ is exactly uniform over $B_S(q,r)$ and does not depend on $R_{\textup{setup}}$.
\end{proof}

Let $\mathcal{A}_{\text{fair}}$ be the Exact Fair ANN algorithm given by \autoref{thm:fair_lsh_prior}.
We now show that fairness implies adversarial robustness. Intuitively, conditioning on success up to round $i-1$ preserves independence between the
adversary’s next query and the preprocessing randomness, allowing us to apply the per-query LSH guarantee at every round.

\begin{claim}[Fairness Implies Robustness]
\label{claim:fairness-robustness}
The algorithm $\mathcal{A}_{\text{fair}}$ is $\tfrac{1}{n}$-adversarially robust for $Q$ adaptive queries.
\end{claim}

\begin{proof}
Without loss of generality, we can consider the adversary as deterministic.
Let $F_i$ denote the event that the algorithm fails (i.e., the candidate set does
not contain the full $r$-ball) at round $i$, and let
$S_i = \bigcap_{j=1}^i \neg F_j$ denote the event of success up to round $i$.

We first prove by induction on $i$ that, conditioned on $S_{i-1}$, the transcript of interactions
$T_{i-1}=(q_1,a_1,\dots,q_{i-1},a_{i-1})$ is independent of the setup randomness
$R_{\textup{setup}}$. For $i=1$, the transcript is empty and the claim is immediate.
Assume the claim holds for $i-1$.
Conditioned on $S_{i-1}$, exact fairness implies that the distribution of
$a_{i-1}$ given $(T_{i-2},q_{i-1})$ does not depend on $R_{\textup{setup}}$.
Hence extending the transcript from $T_{i-2}$ to $T_{i-1}$ preserves independence
from $R_{\textup{setup}}$. Since the adversary is deterministic, the next query
$q_i = \mathcal B(T_{i-1})$ is a function of $T_{i-1}$ and therefore is also
independent of $R_{\textup{setup}}$ conditioned on $S_{i-1}$.

We can now apply the per-query guarantee of
Theorem~\ref{thm:fair_lsh_prior} to $q_i$, yielding
\[
\Pr[F_i \mid S_{i-1}] \le \frac{1}{nQ}.
\]
Finally, by the union bound for sequential events,
\[
\Pr\!\left[\bigcup_{i=1}^Q F_i\right]
\le \sum_{i=1}^Q \Pr[F_i \mid S_{i-1}]
\le \frac{1}{n}.
\]
\end{proof}

\paragraph{Updates to the dataset $S$}
The fair ANN data structure also supports oblivious dynamic updates.
An insertion is handled by hashing the new point into each of the $L$ LSH tables
using the fixed preprocessing hash functions.
A deletion is handled by removing the point from the corresponding buckets,
which can be implemented either by maintaining bucket pointers or by querying
the data structure with the point itself and deleting it from the resulting
candidate lists. Each update touches the same set of hash tables and incurs the same asymptotic
cost as a query. Consequently, the expected update time matches the expected query time up to
constant factors.

Both operations depend only on the preprocessing randomness and do not use
transcript-dependent randomness.
Moreover, the exact neighborhood sampling procedure is unchanged and continues
to use fresh per-query randomness.
Therefore, exact fairness and the robustness analysis remain valid under
any sequence of oblivious updates.

% \paragraph{Discussion}
% A critical direction for future inquiry is the relaxation of the \textit{exact fairness} requirement to \textit{approximate fairness}. The current theoretical framework relies on the strong assumption that the algorithm returns a point sampled from a perfectly uniform distribution over the set of near neighbors, which statistically decouples the output from the internal randomness.[1] However, practical implementations of fair sampling, particularly those utilizing Locality-Sensitive Hashing (LSH), often achieve only approximate uniformity, where the probability of sampling a specific neighbor lies within a multiplicative factor of the ideal uniform probability.[2] Future work should rigorously analyze whether $(\epsilon, \delta)$-approximate fairness is sufficient to bound the information leakage over a sequence of $Q$ adaptive queries. If robustness can be maintained under this relaxed condition---potentially with a degradation in the security parameter proportional to the approximation error---it would allow for the deployment of significantly more efficient sampling algorithms that avoid the heavy rejection sampling overhead required for exact fairness.

% \begin{remark}
%     As discussed in the introduction, this argument extends to any algorithm which is required to answer adaptive queries by picking from a discrete set of candidate values.
% \end{remark}

\section{Assumption-Free Robust Searching via Bucketing}
A major limitation of the fair algorithm is that it only works efficiently when the dataset does not induce a high density ratio, which is not guaranteed if $S$ is picked by the adversary. Ideally, we aim to obtain sublinear algorithms that work without any assumptions on the input dataset. To do this, we introduce a search-to-decision framework:

\subsection{Weak Decision ANN}
\begin{definition}[\textsc{Weak-Decision-ANN}]
Consider the metric space $\mathcal{M}$ and let $S \subseteq U$ with $|S| = n$ be an input point dataset. Let $r > 0$, $c > 1$ be two parameters and $q \in \mathcal{M}$ be an adaptively chosen query. If $B_S(q,r) \neq \emptyset$, then we should answer \textbf{1}. If $B_S(q,cr) = \emptyset$, we must answer \textbf{0}. In any other case, any answer is acceptable. Let $D(S,q,c,r) \subseteq \{0,1\}$ denote the set of correct answers to this weak decision problem for dataset $S$, parameters $c,r$ and query $q$.
\end{definition}
Let $\mathcal{A}$ be an algorithm for solving the weak decision ANN problem, though not necessarily robustly. We can design an adversarially robust decider $\mathcal{A}_{\text{dec}}$ by using $\mathcal{A}$, while only increasing the space by a factor of $\sqrt{Q}$. Adhering to the framework of \citet{hassidim2022adversarially}, we maintain $L= \widetilde{\Theta}(\sqrt{Q})$  copies of the data structures $\cD_1,...,\cD_L$ generated by $\mathcal{A}$ using $L$ independent random strings, and then for each query $q$ we combine the answers of $\mathcal{A}$ privately. As opposed to the original framework of \citep{hassidim2022adversarially}, we do not need to use a private median algorithm, which simplifies the analysis. To keep the query time small, we utilize privacy amplification by subsampling $k \ll L$ copies per query (\autoref{lem:privacy_ampl_subsampling}). Our algorithm appears below as \autoref{alg:robust-decider}.

\begin{algorithm}[h!]
\caption{The robust decider $\mathcal{A}_{\text{dec}}$ (based on an oblivious decider $\mathcal{A}$)}
\label{alg:robust-decider}
\begin{algorithmic}[1]
    \State \textbf{Inputs:} Random string $R = \sigma_1\circ \sigma_2\circ\cdots \sigma_L$.
    \State \textbf{Parameters}: Number of queries $Q$, number of copies $L$, number of sampled indices $k$.
    \vspace{1mm}
    \State Receive input dataset $S \subseteq U$ from the adversary, where $n = |S|$.
    \vspace{0.5mm}
    \State Initialize $\cD_1,...,\cD_L$ where $\cD_i \gets \mathcal{A}(S)$ on random string $\sigma_i$.
    \vspace{1mm}
    \For{$i = 1$ to $Q$}
        \State Receive query $q_i$ from the adversary.
        \State $J_i\gets$ Sample $k$ indices in $[L]$ with replacement.
        \vspace{1mm}
        \State Let $a_{ij} \gets \mathcal{D}_j(q_i) \in \{0,1\}$ and $N_i := \frac{1}{k}\left|\{j \in J_i\mid a_{ij} = 1\}\right|$.\label{line: robust decider: set statistic}
        \vspace{0.75mm}
        \State Let $\widehat{N}_i = N_i + \text{Lap}\left(\frac{1}{k}\right)$.\label{line: robust decider: Laplace mechanism}
        \vspace{0.75mm}
        \State Output $\mathbbm{1}[\widehat{N}_i > \frac{1}{2}]$
    \EndFor
\end{algorithmic}
\end{algorithm}

To analyze this algorithm, we argue that for all $i \in [Q]$, at least $\frac{8}{10}$ of the $k$ answers $a_{ij}$ are correct, even in the presence of adversarially generated queries. To do this, we first need to show that the algorithm is $(\varepsilon,\delta)$-differentially private with respect to the input random strings $R$ which are used to specify the chosen LSH functions, where $\delta$ is the desired failure probability and $\varepsilon$ is an appropriately picked constant. As a result, if we set $L = \Theta(\varepsilon^{-1}\log^{1.5}(1/\delta)\cdot \sqrt{2Q})$ and $k = \Theta(\log(Q/\delta))$ we obtain a robust decider that succeeds with probability at least $1-\Theta(\delta)$. Our analysis (\autoref{thm:robust-decider-thm}) is included in full in \autoref{appendix:decider-proof}. It uses technical tools from Differential Privacy to adapt the robustification argument of \cite{ben2022framework} for our setting. The main difference between their analysis and ours is that we can avoid using a black-box private median algorithm; the Laplace mechanism suffices for our needs.

\subsection{Bucketing-Based Search}
To perform the final search, we partition our point dataset $S$ into $n^{1-\alpha}$ segments, for $\alpha \in (0,1)$. We then instantiate a copy $\cA_i \equiv \cA_{\text{dec}}$ of $\cA_{\text{dec}}$ in each segment. When a query comes in, we forward it to each $\cA_i$ and if some segment answers \textbf{1}, we perform an exhaustive search in the segment to find a point in $B_S(q,cr)$.

\begin{algorithm}[h]
\caption{Robust ANN via Weak Decisions and Bucketing}
\label{alg:bucketization-lsh}
\begin{algorithmic}[1]
    % \State \textbf{Inputs:} Number of segments $k \in \ZZ_+$, ANN parameters $c,r > 0$.
    \State \textbf{Parameters}: Error probability $\delta > 0$, number of queries $Q$
    \State Partition point set $S$ arbitrarily into $\kappa=n^{1-\alpha}$ segments $L_1,...,L_\kappa$ of size $n/\kappa$.
    \State Initialize $\kappa$ independent copies $\cA_1,...,\cA_\kappa$ of $\cA_{\text{dec}}$, each with $\delta' = \delta / \kappa$
    \vspace{0.5mm}
    \For{$i = 1$ to $Q$}
        \State Receive query $q_i$ from the adversary. 
        \For{$j = 1$ to $\kappa$ where $A_j(q_i) = 1$}
            \For{$p \in L_j$}
                \If{$p \in B_S(q,cr)$}
                    \State \textbf{Output} $p$ and proceed to the next query.
                \EndIf
            \EndFor
        \EndFor
        \State Output $\perp$ and proceed to the next query.
    \EndFor
\end{algorithmic}
\end{algorithm}

\begin{lemma}
Algorithm \ref{alg:bucketization-lsh} is a $\delta$-adversarially robust algorithm for the ANN problem.
\end{lemma}
\begin{proof}
Suppose the adversary has won. The algorithm can only make a mistake when all the data structures reply with \textbf{0}, even though there is a point $p \in B_S(q,r)$. Consider the segment $L_i$ for which $p \in L_i$, and examine it in isolation. Because all the copies of Algorithm $\cA_{\text{dec}}$ are initialized independently from each other, the adversary successfully induces $\cA_i$ to make a mistake, which by assumption happens with probability at most $\kappa\delta' = \delta$ over all the segments via union bound.
\end{proof}

Now we proceed to the proof of \autoref{thm:robust-ann-bucketing}:
\begin{proof}[Proof of \autoref{thm:robust-ann-bucketing}]
We create $n^{1-\alpha}$ segments, each containing $n^\alpha$ points. Recall that a single copy of Algorithm $\cA_{\text{dec}}$ takes $\widetilde{O}(n^{1+\rho}\sqrt{Q})$ pre-processing time and space, and $\widetilde{O}(n^\rho)$ time and space per query. Each copy $\cA_i$ runs on $n^\alpha$ points, so for pre-processing, our algorithm uses 
\begin{align*}
\wO\left(n^{1-\alpha} (n^\alpha)^{1+\rho}\sqrt{Q}\right)
=\wO\left(n^{1+\alpha\rho}\sqrt{Q}\right)
\end{align*}
bits of space for creating $n^{1-\alpha}$ copies $\cA_1,...,\cA_L$. On the other hand, to process a single query the algorithm uses
\begin{align*}
\wO\left(n^{1-\alpha}\cdot (n^\alpha)^\rho + n^\alpha\right) 
=\wO\left(n^{1-\alpha+\alpha\rho} + n^\alpha\right)
\end{align*}
time. To balance the summands in the query complexity term, we set 
\begin{align*}
n^{1-\alpha+\alpha\rho} = n^\alpha \implies \alpha = \frac{1}{2-\rho}
\end{align*}
This proves \autoref{thm:robust-ann-bucketing} and concludes the analysis of Algorithm \ref{alg:bucketization-lsh}. 
\end{proof}

The biggest advantage of our algorithm is that it does not make any assumptions on the input dataset. However, it achieves sublinear query time as $\frac{1}{2-\rho} < 1$ when $\rho < 1$. Furthermore, the space complexity of our algorithm for small values of $Q$ is superior to the space complexity of even the oblivious ANN algorithm that has space complexity $n^{1+\rho}$. 

Finally, updates are easily implemented in our framework as well, with insertions and deletions being performed in every LSH copy we maintain. We build $L = \widetilde{O}(\sqrt{Q})$ LSH data structures, each of which has $n^{\rho/(2-\rho)}$ hash functions to update, which yields the final update time complexity. We also need to update our buckets, but we only need to update one bucket at a time as the partitioning over the dataset can be arbitrary. For more efficient performance, in practice we could perform updates lazily, only modifying the data structures we use just-in-time.

% \section{Improvements for Robust and Fair ANN via Concentric Annuli}
% \label{sec:concentric-lsh}
% In this section we present a different method for creating a purely sublinear query time algorithm that beats that $n^{1/(2-\rho)}$ query runtime of \autoref{thm:robust-ann-bucketing}. 
\section{Relaxed Fair ANN via Concentric LSH Annuli}
\label{sec:fairness}
As a warm-up, we first present an algorithmic improvement to \autoref{thm:fair_lsh_prior} for fair ANN, removing the dependency on the ratio $\frac{n(q,cr)}{n(q,r)}$ which could grow as big as $n$ in the query time. We achieve purely sublinear time for a relaxed fairness guarantee:
\begin{definition}[Relaxed Fairness in ANN]
\label{def:relaxed-fairness}
Let $S$ be the input dataset and $q \in \mathcal{M}$ be a query point. If $B_S(q,r) \neq \emptyset$, the algorithm aims to output some point chosen uniformly at random, independently of past queries, from $B_S(q,r')$, where $r' \in [r,cr]$ is a random variable depending on $q$ and $S$. Otherwise, if $B_S(q,r) = \emptyset$, the algorithm can either answer $\perp$ or output a uniformly random point from $B_S(q,r')$ with $r' \in (r,cr]$. 
\end{definition}

Consider the following sequence of radii between $r$ and $cr$, interspersed so that the ratio between two consecutive ones is constant: $r_0 = r,r_1,...,r_{k-1},r_k=cr$ are defined as $r_i = c'\cdot r_{i-1}$ for $i \in \{1,...,k\}$, where $c' = \sqrt[k]{c}$. We create $k$ instances of $\mathcal{A}_{\text{fair}}$, where the $i$-th instance is initialized with parameters $(c',r_k)$. We run each instance to output a point uniformly from $B_S(q,r_i)$. Letting $\rho(c')$ be the LSH-$\rho$ constant depending on the new approximating parameter $c'$, if we observe an instance running for longer than $100\cdot dn^{\max\{\rho(c'),1/k\}}$ timesteps, we stop the execution and switch to the next instance.
\begin{claim}
\label{claim:telescoping}
Consider a query $q$ and suppose $B_S(q,r) \neq \emptyset $. There exists $i \in \{0,...,k-1\}$ such that:
\begin{align}
    \frac{n(q,r_{i+1})}{n(q,r_i)} \leq n^{\frac{1}{k}}
\end{align}
\end{claim}
\begin{proof}
Since $n(q,r) \geq 1$ it also holds that $n(q,r_i) \geq 1$ for all $i \in \{0,...,k\}$. Suppose that for all $i \in \{0,...,k-1\}$ it holds that $\frac{n(q,r_{i+1})}{n(q,r_i)} > n^{\frac{1}{k}}$. Then, via a telescoping product we arrive at a contradiction:
\begin{align}
\frac{n(q,cr)}{n(q,r)} &= \frac{n(q,r_1)}{n(q,r_0)}\cdot \frac{n(q,r_2)}{n(q,r_1)}\cdots\frac{n(q,r_{k-1})}{n(q,r_{k-2})}\cdot\frac{n(q,cr)}{n(q,r_{k-1})}>\left(n^{\frac{1}{k-1}}\right)^{k-1} = n\qedhere
\end{align}
\end{proof}
\autoref{claim:telescoping} shows that if $B_S(q,r) \neq \emptyset$ we output a uniformly sampled point from some $B_S(q,r_i)$, where $r_i$ is a random variable $R$ depending on $S,q$ and our algorithm's randomness. On the other hand, if $B_S(q,r) = \emptyset$, we either output $\perp$ if all the copies $\mathcal{D}_i$ time-out, or a uniformly sampled point from some sphere $B_S(q,r_i)$. In either case, we enjoy the relaxed fairness guarantee of \autoref{def:relaxed-fairness}. For the runtime, our algorithm takes space $O(kn^{1+\rho(c')})$ for pre-processing, and time $\widetilde{O}(dk\cdot \max\{n^{\rho(c')},n^{1/k}\})$ for answering each query. Choosing $k$ that minimizes this exponent we get the following theorem:
\begin{theorem}
\label{thm:fairness-concentric}
There exists an algorithm for solving the relaxed fair $(c,r)$-ANN problem that uses $\widetilde{O}(dn^{\beta})$ time per query and $\widetilde{O}(n^{1+\rho(c^{1/\bar{k}})})$ space for pre-processing, where $\bar{k}=\arg\min\limits_{k\in \mathbb{Z}}\max\{\rho(c^{1/k}), 1/k\}$ and $\beta = \max\{\rho(c^{1/\bar{k}}),1/\bar{k}\}$.
\end{theorem}

Solving for $\beta$ is metric space dependent. For the hypercube, we can use $\rho(c) = \frac{1}{2c-1}$ and a back-of-the-envelope calculation yields $k = \Theta(\frac{\log c}{\log\log c})$. To nail down the constants precisely, we pick:
    $$
    \beta=\min\left\{\max\left\{\frac{1}{\lfloor k^* \rfloor},\frac{1}{2c^{1/\lfloor k^* \rfloor}-1}\right\},\max\left\{\frac{1}{\lceil k^* \rceil},\frac{1}{2c^{1/\lceil k^*\rceil}-1}\right\}\right\}
    $$
    with our algorithm having runtime $\widetilde{O}(dn^\beta)$ and space complexity $\widetilde{O}(\sqrt{Q}\cdot n^{1+\beta})$. For instance, if $c = 4$ we have $k^* = 2.48$, so $\beta = 1/3$, while for $c = 10$ we have $k^* \approx 3.15$ so $\beta=1/3$. We plot the solutions for $\beta$ for $c\in[2,100]$ in \autoref{fig:beta-solutions}, both for the Hamming distance and $\ell_2$ distance metrics. Note that $\beta \to 0$ as $c\to\infty$.
\begin{figure}[h]
    \centering
    \begin{subfigure}[T]{0.48\textwidth}
        \centering
        \includegraphics[width=\textwidth]{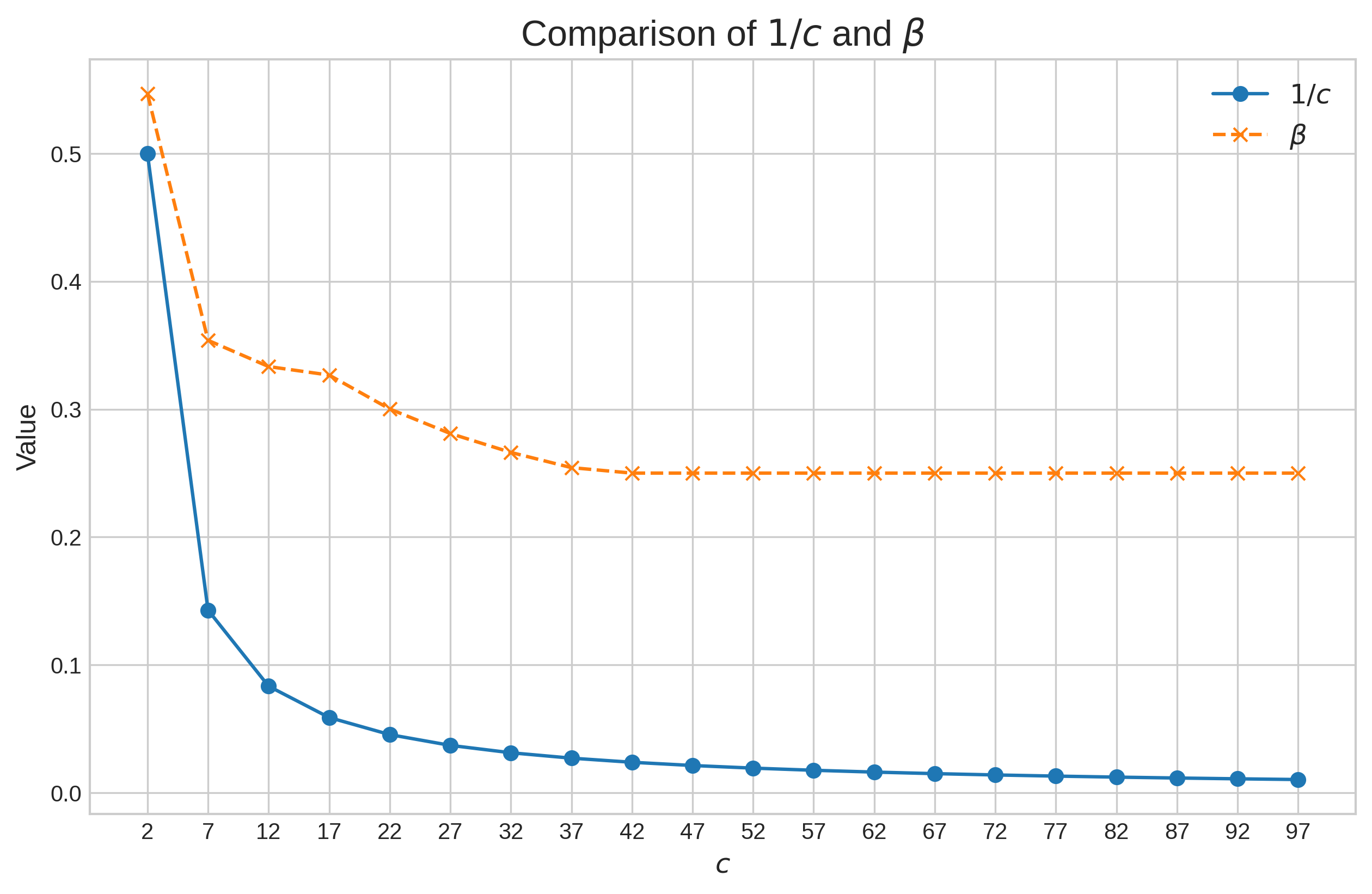}
        % \caption{Values of $\beta$ for the Hamming Cube.}
    \end{subfigure}
    \hfill
    \begin{subfigure}[T]{0.48\textwidth}
        \centering
        \includegraphics[width=\textwidth]{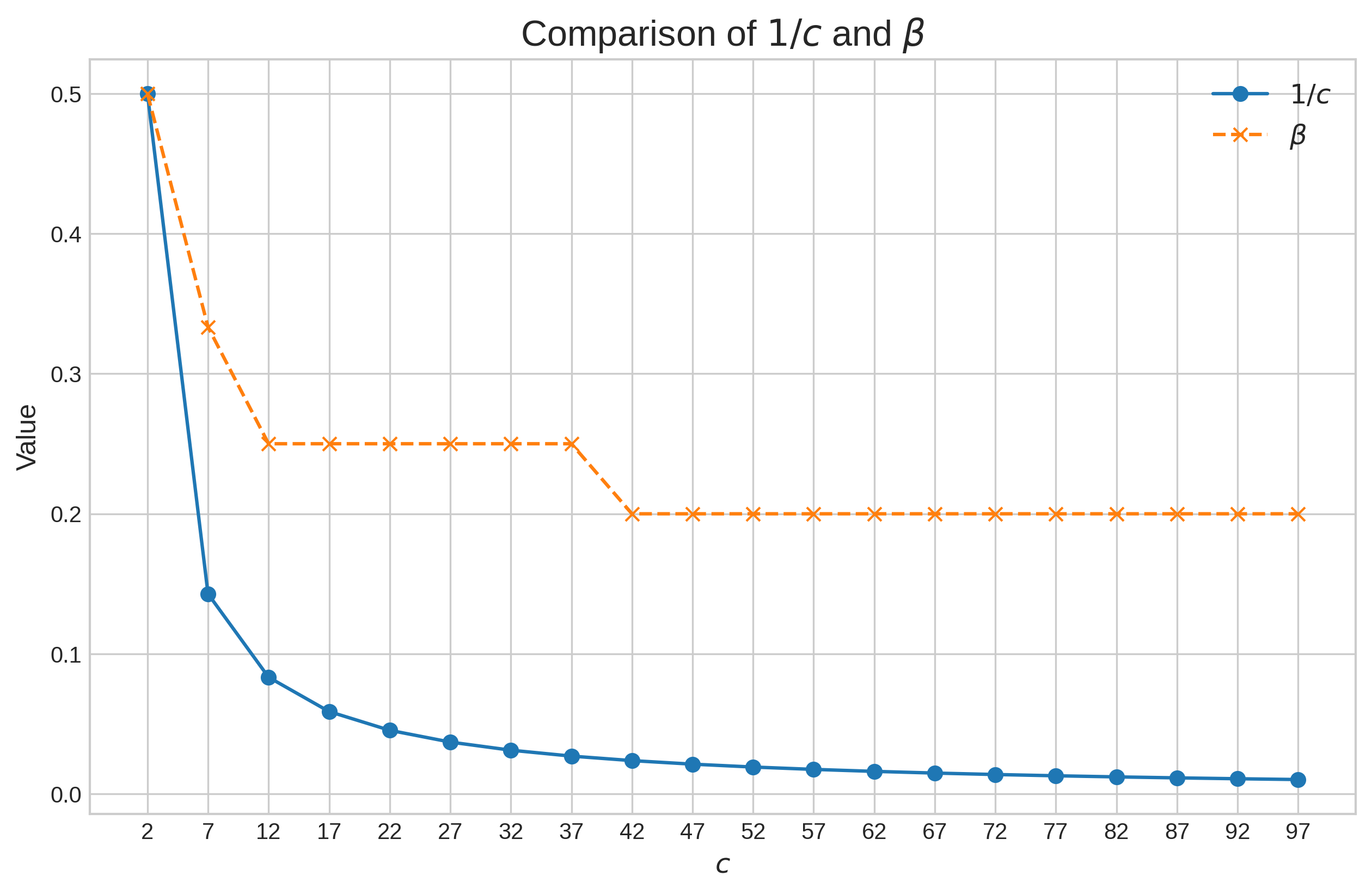}
        % \caption{Values of $\beta$ for $\ell_2$ LSH}
    \end{subfigure}
     \caption{Solutions for $\beta$ for different values of $c$ in the hypercube (left) and $\ell_2$ (right) domains.}
    \label{fig:beta-solutions}
    % \vspace{-5mm}
\end{figure}

\section{Robust ANN Improvements}
We now combine our concentric annuli technique with fair ANN to develop a more efficient and robust algorithm. The core idea is to use a set of parallel instances to dynamically identify an annulus that is likely to yield a fast solution for a given query, and then carefully release this information to maintain adaptive security.

We again partition the space into $k \ge 1$ concentric annuli $(r_{i-1}, r_i]$, where $r_0=r$ and $r_i = c' \cdot r_{i-1}$ for $c' = \sqrt[k]{c}$. For each annulus $i$, we instantiate two independent copies of the base algorithm: a \textit{testing} instance $\mathcal{A}_i \gets \mathcal{A}_{\text{fair}}(c', r_{i-1})$ and a held-out \textit{execution} instance $\mathcal{A}^{(i)}_{\text{fair}} \gets \mathcal{A}_{\text{fair}}(c', r_{i-1})$. Our goal is to find an annulus whose algorithm terminates within a pre-determined runtime threshold with large probability. We formalize this notion as follows.
\begin{definition}[Good Annuli]
Let $T_i$ be the random runtime of the testing instance $\mathcal{A}_i$. The $i$-th annulus is a \textbf{good annulus} if its probability of fast termination, $p_i$, is high:
$$
p_i:=\Pr[T_i\leq 4d(n^{\rho(c')}+n^{1/k})\log(nQ)] \geq 0.999
$$
\end{definition}

Upon receiving a query $q$, we estimate each probability $p_i$ with an additive error of at most $\eta$ by observing the fraction of $\Theta(\eta^{-2}\log(kQ/\delta))$ independent sub-trials of $\mathcal{A}_i$ that halt within the time bound. Let $\hat{p}_i$ be this empirical estimate for $p_i$. We identify candidate annuli with an indicator vector $\hat{\alpha} \in \{0,1\}^k$, where: $\hat{\alpha}_i = \mathbbm{1}[\hat{p_i} \geq 0.998]$.

With high probability, $\hat{\alpha}_i = 1$ implies that annulus $i$ is good. A similar argument to \autoref{claim:telescoping} guarantees that at least one good annulus must exist. We therefore find the first index $i^*$ for which $\hat{\alpha}_{i^*} = 1$ and run the corresponding execution instance $\mathcal{A}_{\text{fair}}^{(i^*)}$ to completion. This approach yields a solution in $\widetilde{O}(d(n^{\rho(c')} + n^{1/k}))$ time with probability at least $0.998$

To ensure robustness, the release of the vector $\hat{\alpha}$ (and thus the choice of $i^*$) must not reveal information about the internal randomness of our algorithm instances. We therefore use the DP-based robustification framework of \citet{hassidim2022adversarially} to release $\hat{\alpha}$ privately. While this increases the space complexity by a factor of $\sqrt{Q}$, it allows us to achieve a considerably better, assumption-free, and purely sublinear query time. \autoref{alg:robust_ann_improved} presents the details of our approach.

To tackle updates, we again simply count the number of LSH data structures we maintain and scale by the number of hash functions (see \autoref{table:comparison})

\begin{algorithm}[h]
\caption{Improved Robust ANN Search }
\label{alg:robust_ann_improved}
\begin{algorithmic}[1]
\State \textbf{Input:} Query $q \in \mathcal{M}$, parameters $c, r, k \geq 1$ and $\delta \in (0,0.0025)$
\Procedure{Initialize}{}
    \State Let $c' \gets \sqrt[k]{c}$ and $r_0 \gets r$.
    \State Let $\eta = 0.001$, $m = \eta^{-2}\log(Qk/\delta)$ and $L=2400\log^{1.5}(1/\delta)\sqrt{2Q}$.
    \For{$i = 1, \dots, k$} 
        \State Let $N=m\times L$ and $r_i \gets c' \cdot r_{i-1}$.  \Comment{Testing Instance Grid}
        \State Instantiate $N$ copies $\mathcal{A}_{i,j_m,j_L} \gets \mathcal{A}_{\text{fair}}(c', r_{i-1})$ for $(j_m,j_L) \in [m]\times [L]$.
        \State Instantiate $\mathcal{A}_{\text{fair}}^{(i)} \gets \mathcal{A}_{\text{fair}}(c', r_{i-1})$ \Comment{Execution Instances}
    \EndFor
\EndProcedure
\Procedure{Query}{$q$}
    \For{$i \in \{1, \dots, k\}$}
        \State Let $S_{\text{trunc}} \gets 4d(n^{1/k}+n^{\rho(c')})\log(nQ)$ \Comment{Let $p_j \gets \Pr[T_j < S_{\text{trunc}}]$.}
        \State $J_i\gets$ Sample $s=\frac{2}{\eta}\log(Qk/\delta)$ indices $(j_m,j_L)\in [m]\times [L]$ with replacement.
        \For{$j:=(j_m, j_L) \in J_i$} 
            \State \Comment{$T_{i,j_m,j_L}$ is the execution time of testing instance $A_{i,j_m,j_L}$.}
            \State Let $\widetilde{p}_{ij} \gets \mathbbm{1}[T_{i,j_m,j_L} < S_{\text{trunc}}]$. 
        \EndFor
        \State Let $\hat{p}_i \gets \frac{1}{s}\sum\limits_{j\in J_i} \widetilde{p}_{ij} + \text{Lap}(\frac{1}{s})$.
        \State Set $\hat{a}(q)_i \gets \mathbbm{1}[\hat{p}_i \geq 0.998]$.
    \EndFor
    \If{$\hat{a}(q) = \vec{0}$}
        \State \textbf{return} $\perp$
    \Else
        \State $i^* \gets \min \{ i \in \{1, \dots, k\} \mid a(q)_i = 1 \}$ \Comment{Find most significant bit index}
        \State \textbf{return} $\mathcal{A}_{\text{fair}}^{(i^*)}(q)$
    \EndIf
\EndProcedure
\end{algorithmic}
\end{algorithm}

\subsection{Analysis}
 First we show that the vector $\hat{a}(q)$ is produced robustly.

\begin{figure}[h]
    \centering
    \begin{tikzpicture}[scale=1]
        % Outer annulus (green fill)
        \fill[green!20] (0,0) circle (3cm);
        \fill[blue!20] (0,0) circle (3cm) -- (0,0) circle (2cm);
    
        % Middle annulus (blue fill)
        \fill[blue] (0,0) circle (2cm);
        \fill[green!20] (0,0) circle (2cm) -- (0,0) circle (1cm);
    
        % Inner annulus (green fill)
        \fill[green!20] (0,0) circle (1cm);
    
        % Borders for each annulus
        \draw (0,0) circle (3cm);
        \draw (0,0) circle (2cm);
        \draw (0,0) circle (1cm);
        \draw (0,0) circle (4cm);

        % Labels
        \node at (2.5, 0) {$r_{i+1}$}; % Label for the outer green circle
        \node at (1.5, 0) {$r_{i}$}; % Label for the blue circle

    \end{tikzpicture}
    \caption{Our concentric LSH construction. In green lies the set $B_S(q,r_i)$, and blue represents the annulus that extends to $B_S(q,r_{i+1})$}
    \label{fig:concentric_lsh}
\end{figure}
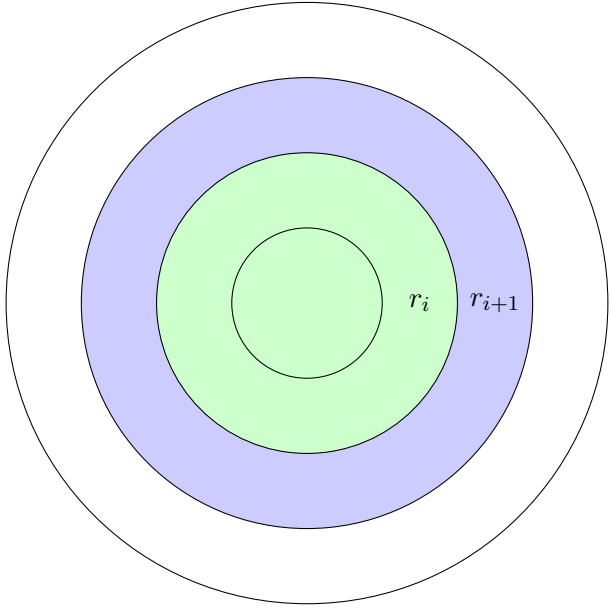

\begin{lemma}
\label{lemma:p_j_approx}
    With probability at least $1-\Theta(\delta)$ over all queries and annuli, the vectors $\hat{p}$ computed by \autoref{alg:robust_ann_improved} are such that:
    \begin{align*}
        ||\hat{p}-p||_\infty \leq \eta
    \end{align*}
    This holds despite the adversary's action to establish the opposite.
\end{lemma}
\begin{proof}
Our argument mimics the proof of \autoref{thm:robust-decider-thm} in that it invokes the robustification framework of \citet{hassidim2022adversarially}. 

First, producing $\hat{p}$ over $Q$ timesteps is $(\varepsilon,\delta)$-private with respect to the random strings our algorithm uses. This follows from the application of the Laplace mechanism with sensitivity $\frac{1}{s}$ (\autoref{thm:laplace-mechanism}) and from privacy amplification by subsampling (\autoref{lem:privacy_ampl_subsampling}). Furthermore, if each $\hat{p}_i$ is produced privately with respect to the input randomness, we can invoke the parallel composition theorem of DP (\autoref{thm:parallel-composition}) to show that the entire release of $\hat{p}$ is private without additional cost to the privacy parameters. As we know, privacy with respect to the input randomness implies robustness, so we now have to calculate the cost of privacy in our approximation algorithm.

It suffices to argue that for a fixed $i \in [k]$ we have $|\hat{p}_i - p_i| \leq \eta$. Fixing some query and $i \in [k]$ we know by the triangle inequality and a standard Chernoff bound (since $m = \Omega(\eta^{-2}\log (Qk/\delta))$) that:
$$
|p_i - \hat{p}_i| \leq \left|p_i-\frac{1}{s}\sum\limits_{j \in J_i}\widetilde{p}_i\right|+\left|\frac{1}{s}\sum\limits_{j \in J_i}\widetilde{p}_{ij} - \widehat{p}_i\right| \leq \frac{\eta}{2} + \left|\frac{1}{s}\sum\limits_{j \in J_i}\widetilde{p}_{ij} - \widehat{p}_i\right|
$$
The latter term of the above sum is the error incurred via the privatization process. We can bound it by using our known bound on the magnitude of Laplacian noise (\autoref{lem:bounded max of laplace distribution}):
\begin{align*}
   \left|\frac{1}{s}\sum\limits_{j \in J_i}\widetilde{p}_{ij} - \widehat{p}_i\right| = \left|\frac{1}{s}\sum\limits_{j \in J_i}\widetilde{p}_{ij} - \frac{1}{s}\sum\limits_{j \in J_i}\widetilde{p}_{ij} + \text{Lap}\left(\frac{1}{s}\right)\right| \leq \frac{1}{s}\cdot\frac{\eta}{2}\cdot s=\frac{\eta}{2}
\end{align*}
This happens with probability at least $1-e^{-\eta s} \geq 1-\frac{1}{\text{poly}(kQ/\delta)}$. Taking a union bound over $k$ annuli and $Q$ queries establishes the lemma.
\end{proof}
\begin{corollary}
    Since \autoref{alg:robust_ann_improved} generates vector $\vec{a} \in \{0,1\}^k$ by post-processing, this also implies that $a$ is generated robustly.
\end{corollary}

Next, we argue that the output point of the algorithm is produced correctly and within the claimed runtime and space complexity.
\begin{theorem}
    \autoref{alg:robust_ann_improved} is a $(\delta+\frac{1}{\text{poly}(n)})$-robust $(c,r)$--ANN algorithm that uses pre-processing space $\widetilde{O}(\sqrt{Q}\cdot n^{1+\rho(c')})$, where $c' = c^{1/\bar{k}}$ and $\bar{k} = \arg\min_{k \in \mathbb{Z}_{\geq 1}}\left(\beta:=\max\{\rho(c^{1/k}), 1/k\}\right)$. Each query takes $\widetilde{O}(dn^{\beta})$ time with probability at least $0.998$.
\end{theorem}
\begin{proof}
    We maintain $\Theta(k\log^{2.5}(kQ)\sqrt{Q})$ testing instances, as well as $k$ execution instances. This means our total space complexity is:
    \[
    \Theta\left(k\log^{2.5}(kQ/\delta)\sqrt{Q}\cdot n^{1+\rho(c')} + kn^{1+\rho(c')}\right) = \widetilde{O}(k\sqrt{Q}\cdot n^{1+\rho(c')}).
    \]
    For the query runtime, suppose $B_S(q,r) \neq \emptyset$. As we argued in \autoref{claim:telescoping}, there must exist some annulus $\ell$ for which the density ratio is at most $n^{1/k}$. For that annulus, \autoref{thm:fair_lsh_prior} implies that:
    \begin{align*}
        \Pr\left[T_\ell <4(n^\rho+n^{1/k})\log(nQ)\right] \geq 1-\frac{1}{\text{poly}(n)} \gg 0.999
    \end{align*}
    Therefore, there always exists a good annulus when $B_S(q,r) \neq \emptyset$. 
    
    By \autoref{lemma:p_j_approx} we have that a good annulus will, with high probability be captured by \autoref{alg:robust_ann_improved}. Conversely, if $\hat{a}(q)_i = 1$, then $p_i \geq 0.999-\eta = 0.998$. As a result, if $i^*$ is the MSB of $a$, the corresponding execution instance $\mathcal{A}_{\text{fair}}^{(i^*)}$ runs in time $O(dn^{\beta})$ with probability at least $0.998$. Overall, to process one query, we run all $ks = O(\log(Qk))$ truncated copies of $\mathcal{A}_{ij}$. Thus, our algorithm takes $O(d\log (nQ)\cdot n^{\beta})$ per query, as initially claimed.

    Finally, to argue robustness, we know from \autoref{lemma:p_j_approx} that releasing vector $\hat{a}$ is done robustly. Also, \autoref{claim:fairness-robustness} tells us that the held-out execution copy is robust, given that the MSB $i^*$ is produced from $a$ via a fixed function (post-processing). Overall, the output of \autoref{alg:robust_ann_improved} is adversarially robust with probability at least $1-\delta-\frac{1}{\text{poly}(n)}$, accounting for the probability that any of the fair ANN algorithms fail. 
\end{proof}
Finally, we use the median amplification trick to get a high probability guarantee on the runtime. We maintain $t=\Theta(\log(Q/\delta))$ independent copies of the algorithm above and declare the runtime to be the median runtime. A standard Chernoff bound argument shows that the probability any of $Q$ queries failing is at most $\delta/Q$, which by union bound makes the failure probability at most $\delta$, as desired.

\section{Conclusion}
This study presents a series of algorithms for solving ANN against adaptive adversaries. Our approaches, which integrate principles of fairness and privacy with novel data constructions, are efficient and input-independent. Our work raises several intriguing questions for future research: Can we establish time and space lower bounds for robust algorithms? Can the powerful link between fairness and robustness be extended to other domains, like estimation problems? Lastly, a practical implementation of our approach and the nuances it presents in a real system are also important avenues to investigate. 

\section*{Acknowledgments}
The authors would like to thank the organizers of the Summer Program on Sublinear Algorithms at the Simons Institute, where this work was initiated. We are grateful to Sofya Raskhodnikova for early discussions that helped formulate the problem, and to Adam Smith for explaining the recent status of techniques in adaptive data analysis. We also thank Christian Sohler for valuable discussions regarding his paper, and Rathin Desay for his helpful comments.

\bibliography{refs}
\bibliographystyle{apalike}

\newpage
\appendix
% \startcontents[appendix] % Start a new ToC named "appendix"
% \printcontents[appendix]{l}{1}{\section*{Appendix Contents}} % Print it

\section{Notation Table}
\begin{table}[h]
\centering
\renewcommand{\arraystretch}{1.2}
\begin{tabular}{|c|p{10.5cm}|}
\hline
\textbf{Notation} & \textbf{Description} \\ \hline

$n$ & Number of points in the dataset $P \subseteq \mathbb{R}^d$ \\ \hline
$d$ & Dimension of the ambient space \\ \hline
$Q$ & Number of (possibly adaptive) queries \\ \hline
$r$ & Target radius for the $(c,r)$-ANN problem \\ \hline
$c>1$ & Approximation factor \\ \hline
$q$ & A query point \\ \hline
$B_S(q,r)$ & Ball of radius $r$ centered at $q$, intersected with dataset $S$\\ \hline
$n(q,r)$ & Number of dataset points in $B_S(q,r)$ \\ \hline
$\rho(c)$ & LSH exponent for approximation factor $c$ \\ \hline

$A_{\text{dec}}$ & Weak $(c,r)$-ANN decision algorithm \\ \hline
$A_{\text{fair}}$ & Fair ANN algorithm whose output distribution is independent of past queries \\ \hline

$\alpha$ & Bucketing parameter (segment size exponent) \\ \hline
$S_i$ & The $i$-th bucket/segment of the dataset \\ \hline

$k$ & Number of concentric annuli (or subsampling parameter, context-dependent) \\ \hline
$r_i$ & Radius of the $i$-th annulus, $r_i = c^{i/k} r$ \\ \hline
$c'$ & Per-annulus approximation factor, $c' = c^{1/k}$ \\ \hline
$\beta$ & Query-time exponent $\beta = \min_k \max\{\rho(c^{1/k}), 1/k\}$ \\ \hline

$L$ & Number of independent algorithmic copies (for robustness) \\ \hline
$m$ & Number of runtime trials per copy \\ \hline
$T_{i,j}$ & Runtime of the $j$-th copy on annulus $i$ \\ \hline
$S_{\text{trunc}}$ & Runtime truncation threshold \\ \hline

$N_i$ & Fraction of copies that halt within time $S_{\text{trunc}}$ on annulus $i$ \\ \hline
$\hat N_i$ & Noisy estimate of $N_i$ released via the Laplace mechanism \\ \hline

$\varepsilon$ & Differential privacy parameter \\ \hline
$\delta$ & Failure probability parameter / DP parameter \\ \hline

% $t$ & Number of repetitions in median amplification \\ \hline
% $T$ & Per-query runtime upper bound \\ \hline

\end{tabular}
\caption{Summary of notation used throughout the paper.}
\label{table:notation}
\end{table}

\section{Background from Differential Privacy}
\label{appx:dp-review}
Our work leans heavily on results from differential privacy, so we give the necessary definitions and results here.

\subsection{Definition of differential privacy}
\begin{definition}[Differential Privacy]
Let $\mathcal{A}$ be any randomized algorithm that operates on databases whose elements come from some universe. For parameters $\varepsilon > 0$ and $\delta \in [0,1]$, the algorithm $\mathcal{A}$ is $(\varepsilon,\delta)$--differentially private (DP) if for any two neighboring databases $S \sim S'$ (ones that differ on one row only), the distributions on the algorithm's outputs when run on $S$ vs $S'$ are very close. That is, for any $S\sim S'$ and any subset of outcomes $T$ of the output space of $\mathcal{A}$ we have:
    \begin{align*}
        \Pr[\mathcal{A}(S)\in T] \leq e^\varepsilon\cdot \Pr[\mathcal{A}(S')\in T] + \delta
    \end{align*}
\end{definition}

\subsection{The Laplace Mechanism and its properties}
\begin{theorem}[The Laplace Mechanism, \citep{dwork2006calibrating}]\label{thm:laplace-mechanism} Let $f : X^* \rightarrow \mathbb{R}$ be a function. Define its {sensitivity} $\ell$ to be an upper bound to how much $f$ can change on neighboring databases:
\begin{align*}
    \forall S\sim S':\quad |f(S)-f(S')| \leq \ell
\end{align*}
The algorithm that on input $S \in X^*$ returns $f(S) + \text{Lap}\left(\frac{\ell}{\varepsilon}\right)$ is $(\varepsilon, 0)$--DP, where 
$$
\text{Lap}(\lambda;x) := \frac{1}{2\lambda}\exp\bigg(-\frac{|x|}{\lambda}\bigg)
$$
is the Laplace Distribution over $\mathbb{R}$.
\end{theorem}

% As an example, suppose that $S$ contains the medical records of $n$ individuals and $f$ gives the percentage of people with diabetes in the database. The sensitivity is at most $1/n$, because changing one person in the database can only change the fraction of diabetic people by $1/n$. According to the Laplace Mechanism, we can privately release statistics about the percentage of diabetics, by adding noise to the answer proportional to $\text{Lap}\left(\frac{1}{n\varepsilon}\right)$. 
We will make use of the following concentration property of the Laplace Distribution:
\begin{lemma}\label{lem:bounded max of laplace distribution}
For $m \geq 1$, let $Z_1,...Z_{m}\sim \text{Lap}\left(\lambda\right)$ be iid random variables. We have that:
    \begin{align*}
        \Pr\left[\max\limits_{i=1}^m Z_i> \lambda(\ln(m) + t)\right] \leq e^{-t}
    \end{align*}
\end{lemma}

\subsection{Properties of differential privacy}
Differential Privacy has numerous properties that are useful in the design of algorithms. The following theorem is known as ``advanced adaptive composition'' and describes a situation when DP algorithms are linked sequentially in an adaptive way.
\begin{theorem}[Advanced Composition, \citep{dwork2010boosting}]\label{thm:advanced composition}
    Suppose algorithms $\mathcal{A}_1,...,\mathcal{A}_k$ are $(\varepsilon, \delta)$--DP. Let $\mathcal{A}'$ be the adaptive composition of these algorithms: on input database $x$, algorithm $\mathcal{A}_i$ is provided with $x$, and, for $i \geq 2$, with the output $y_{i-1}$ of $\mathcal{A}_{i-1}$. Then, for any $\delta' \in (0,1)$, Algorithm $\mathcal{A}$ is $(\widetilde{\varepsilon},\widetilde{\delta})$--DP with:
    $$\widetilde{\varepsilon} = \varepsilon\cdot\sqrt{2k\ln(1/\delta')}+2k\varepsilon^2\,\text{ and }\,\widetilde{\delta}=k\delta+\delta'$$
\end{theorem}
\noindent There is also a composition theorem concerning situations where the dataset is partitioned:
\begin{theorem}[Parallel Composition]
    \label{thm:parallel-composition}
    Let $f_1,...,f_k$ be $(\varepsilon,0)$-DP mechanisms and $X$ be a dataset. Suppose $X$ is partitioned into $k$ parts $X_1,...,X_k$ and let $f(X) = (f_1(X_1),...,f_k(X_k))$. Then $f$ is $(\varepsilon,0)$-DP.
\end{theorem}
\noindent The next theorem dictates that post-processing the output of a DP algorithm cannot degrade its privacy guarantees, as long as the processing does not use information from the original database.
\begin{theorem}[DP is closed under Post-Processing]\label{thm:post processing DP}
Let $\mathcal{A} : U^n \to Y^m$ and $\mathcal{B} : Y^m \to Z^r$ be randomized algorithms, where $U, Y, Z$ are arbitrary sets. If $\mathcal{A}$ is $(\varepsilon,\delta)$--DP, then so is the composed algorithm $\mathcal{B}(\cA(\cdot))$.
\end{theorem}

The following theorem showcases the power of DP algorithms in learning. 
\begin{theorem}[DP and Generalization, \citep{bassily2016algorithmic, dwork2015preserving}]
\label{thm:dp-generalization}
Let $\varepsilon \in (0,1/3)$ and $\delta \in (0,\varepsilon/4)$. Let $\mathcal{A}$ be a $(\varepsilon,\delta)$--DP algorithm that operates on databases in $X^n$ and outputs $m$ predicate functions $h_i : X \to \{0,1\}$ for $i \in [m]$. Then, if $D$ is any distribution over $X$ and $S$ consists of $n\geq \frac{1}{\varepsilon^2}\cdot \log\left(\frac{2\varepsilon m}{\delta}\right)$ iid samples from $D$, we have for all $i \in [m]$ that:
\begin{align*}
\Pr_{\substack{S\sim D^n\\ h_i\gets \mathcal{A}(S)}}\left[\left|\frac{1}{|S|}\sum\limits_{x\in S}h_i(x)-\underset{x\sim D}{\mathbb{E}}[h_i(x)]\right| \geq 10\varepsilon\right] \leq \frac{\delta}{\varepsilon}
\end{align*}
\end{theorem}
\noindent 
In other words, a privately generated predicate is a good estimator of its expectation under any distribution on the input data. A final property of privacy that we will use is a boosting technique through sub-sampling:
\begin{theorem}[Privacy Amplification by Subsampling, \citep{bun2015differentially, cherapanamjeri2023robust}]
\label{lem:privacy_ampl_subsampling}
Let $\cA$ be an $(\eps,\delta)$--DP algorithm operating on databases of size $m$.  For $n \geq 2m$, consider an algorithm that for input a database of size $n$, it subsamples (with replacement) $m$ rows from the database and runs $\cA$ on the result. Then this algorithm is $(\varepsilon', \delta')$--DP for
\begin{align*}
    \varepsilon' = \frac{6\varepsilon m}{n}\,\text{ and }\,\delta'=\exp\left(\frac{6\varepsilon m}{n}\right)\cdot\frac{4m}{n}\cdot \delta
\end{align*}
\end{theorem}

\subsection{Robustification via Privacy Over Internal Randomness}
A recurring tool used in this work is that we can obtain robustness against adaptive queries by enforcing stability with respect to the algorithm’s internal randomness, rather than the input dataset \citep{hassidim2022adversarially}. Concretely, we view the random coins used by the algorithm as a (hidden) database and require that the observable transcript of interaction be differentially private with respect to changes in this randomness. DP then implies that no adaptively chosen sequence of queries can significantly depend on or overfit to particular random choices made during setup or execution. We explicitly utilize and refine such arguments in the proof of \autoref{appendix:decider-proof}.

\section{Proof of \autoref{thm:robust-decider-thm}}
\label{appendix:decider-proof}
In this section we include a formal analysis of the construction of \autoref{alg:robust-decider}. We prove the following theorem:

\begin{theorem}
\label{thm:robust-decider-thm}
Let $\mathcal{A}$ be an oblivious decider algorithm for ANN that uses $s(n)$ space and $t(n)$ time per query. Let $\varepsilon=0.01$, $\delta < \varepsilon/4$ and suppose we set $L = 24\varepsilon^{-1}\log^{1.5}(1/\delta)\cdot \sqrt{2Q}$ and $k =\log(Q/\delta)$.
Then, the algorithm $\mathcal{A}_{\text{dec}}$ is an adversarially robust decider that succeeds with probability at least $1-\Theta(\delta)$ using $s(n)\cdot \widetilde{O}\left(\sqrt{Q}\right)$ bits of space and $\widetilde{O}\left(t(n)\right)$ time per query. 
\end{theorem}

First, we show that the algorithm is differentially private with respect to its input randomness.

\begin{lemma}
\label{lemma:dec-dp}
Let $\varepsilon=0.01$ and $\delta < \varepsilon/4$. Algorithm $\mathcal{A}_{\text{dec}}$ is $(\varepsilon,\delta)$--DP with respect to the string of randomness ${R}$.
\end{lemma}
\begin{proof}
We analyze the privacy of the algorithm $\mathcal{A}_{\text{dec}}$ given in Algorithm \ref{alg:robust-decider} with respect to the string of randomness $R$, which we interpret as its input. 
Suppose we let 
$$
\varepsilon'= \frac{\varepsilon}{2\sqrt{2Q\ln(1/\delta)}}
$$
For all $i \in [Q]$, we claim that the response to query $q_i$ is $(\varepsilon',0)$--DP with respect to $R$. This is because the statistic $N_i$ defined in Line \ref{line: robust decider: set statistic} of Algorithm \ref{alg:robust-decider} has sensitivity $1/k$ and therefore by \autoref{thm:laplace-mechanism}, after applying the Laplace mechanism in Line \ref{line: robust decider: Laplace mechanism}, we have that releasing $\widehat{N}_i$ is 
$(1,0)$--DP with respect to the strings $R$. The binary output based on % we choose our output by 
comparing $\widehat{N}_i$ with the constant threshold $1/2$ is still  $(1,0)$-DP by post-processing (\autoref{thm:post processing DP}). %From the post-processing theorem of differential privacy, we conclude that releasing our output in each iteration after inspecting $J_i$ is $(1,0)$-DP. 

Since $L \geq 2k$, using the amplification by sub-sampling property (\autoref{lem:privacy_ampl_subsampling}), we get that each iteration is $(\varepsilon',0)$--DP, because for large enough $Q$ we have:
\begin{align*}
\frac{6k}{L} &= \frac{6\cdot 0.01(\log Q+\log\frac{1}{\delta})}{24\log\frac{1}{\delta}\sqrt{2Q\log\frac{1}{\delta}}}\\
&=\frac{6\varepsilon\log\frac{1}{\delta}+6\varepsilon\log Q}{24\cdot \log\frac{1}{\delta}\sqrt{2Q\ln\left(\frac{1}{\delta}\right)}} \tag{Since $\varepsilon = 0.01$} \\
&< \frac{2\varepsilon}{4\sqrt{2Q\ln\left(\frac{1}{\delta}\right)}} \\
&= \varepsilon'
\end{align*}
Finally, by adaptive composition (\autoref{thm:advanced composition}), %the advanced composition theorem of DP tells us that
after $Q$ adaptive steps our resulting algorithm is $(\varepsilon'',\delta)$-DP where:
\begin{align*}
\varepsilon'' = \varepsilon' \sqrt{2Q\ln\left(\frac{1}{\delta}\right)} + Q(\varepsilon')^{2} = \frac{\varepsilon}{2} + \frac{\varepsilon^2}{4\ln\left(\frac{1}{\delta}\right)} \leq \varepsilon
\end{align*}
for $\varepsilon \leq 2\ln\delta^{-1}$, which is satisfied for $\delta \in (0,0.0025)$. Thus, Algorithm $\mathcal{A}_{\text{dec}}$ is $(\varepsilon,\delta)$--DP with respect to its inputs -- the random strings $R$.
\end{proof}
\noindent
Next, we show %what we discussed earlier -
that a majority of the data structures $\mathcal{D}_i$ output accurate verdicts with high probability, even against adversarially generated queries.\\
\begin{lemma}
\label{lem:majority_is_correct}
With probability at least $1-\delta$, for all $i \in [Q]$, at least $0.8L$ of the answers $a_{ij}$ are accurate responses to the decision problem with query $q_i$.
\end{lemma}
\begin{proof}
The central idea of the proof, as it appeared in \citep{hassidim2022adversarially}, is to imagine the adversary $\mathcal{B}$ as a post-processing mechanism that tries to guess which random strings lead $\mathcal{A}$ to making a mistake.

Imagine a wrapper \textit{meta-algorithm} $\mathcal{C}$, outlined as Algorithm \ref{alg:meta-algorithm}, that takes as input the random string $R = \sigma_1\circ r_2\circ\cdots\circ \sigma_L$, which is generated according to some unknown, arbitrary distribution $\mathcal{R}$. This algorithm $\mathcal{C}$ simulates the game between $\mathcal{A}_{\text{dec}}$ and $\mathcal{B}$: It first runs $\mathcal{B}$ to provide some input dataset $S \subseteq U$ to $\mathcal{A}_\text{dec}$, which is seeded with random strings in $R$. Then, $\mathcal{C}$ uses $\mathcal{B}$ to query $\mathcal{A}_\text{dec}$ adaptively with queries $(q_1,...,q_Q)$. At the same time, it simulates $\mathcal{A}_\text{dec}$ to receive answers $a_1,...,a_Q$ that are fed back to $\mathcal{B}$. By \autoref{lemma:dec-dp}, the output $(a_1,...,a_Q)$ is produced privately with respect to $R$, regardless of how the adversary makes their queries. 

At every step $i$, once $\mathcal{B}$ has provided $\vec{q_i} = (q_1,...,q_i)$ and has gotten back $i$ answers $(a_1,...,a_i)$ from $\mathcal{A}_{\text{dec}}$, our meta-algorithm $\mathcal{C}$ \textit{post-processes} this transcript $\{(q_j,a_j)\}_{j=1}^i$ to generate a predicate $h_{\vec{q_i}}:\{0,1\}^*\to\{0,1\}$. This predicate tells which strings $\sigma \in \{0,1\}^*$ lead algorithm $\mathcal{A}$ to successfully answer query prefix $\vec{q_i}$ on input dataset $S$, in the decision-problem regime. More formally:
\begin{align}
    \label{eq:predicate_definition}
    h_{\vec{q_i}}(\sigma) := \bigwedge\limits_{1\leq j\leq i}\left\{\mathcal{A}(\sigma)(S,q_j) \in D(S,q_j,c,r)
    \right\}
\end{align}
Note that this definition captures the intermediate case in which any answer of the algorithm is considered correct.

\begin{algorithm}[ht]
\caption{The meta-algorithm $\mathcal{C}$, ran for $i$ steps}
\label{alg:meta-algorithm}
\begin{algorithmic}[1]
    \State \textbf{Inputs:} Random string $R = \sigma_1\circ \sigma_2\circ\cdots \sigma_L$, descriptions of Algorithms $\mathcal{A}_{\text{dec}}$ and $\mathcal{B}$.
    \vspace{1mm}
    \State Simulate $B$ to obtain a dataset $S \subset U$. 
    \State Initialize $\mathcal{A}_\text{dec}$ with random strings $(\sigma_1,...,\sigma_L)$ and the dataset $S$. 
    \For{$i \in Q$}
        \State Simulate $\mathcal{B}$ to produce a query $q_j$ based on the prior history of queries and answers.
        \State Simulate $\mathcal{A}$ on query $q_j$ to produce an answer.
        \State Compute (via post-processing of query/answer history) predicate $h_{\vec{q_i}}(\cdot)$ from \autoref{eq:predicate_definition}.
    \EndFor
    \State \textbf{Output} $(h_{\vec{q_1}},...,h_{\vec{q_Q}})$. 
\end{algorithmic}
\end{algorithm}

Generating these predicates is possible because $h_{\vec{q_i}}$ only depends on $\vec{q_i}$, which is a substring of the output history that $\mathcal{C}$ has access to. As a result, $\mathcal{C}$ can produce $h_{\vec{q_i}}$ by (say) calculating its value for each value of $R$ exhaustively\footnote{We assume $\mathcal{C}$ has unbounded computational power.}. Because $\mathcal{C}$ is only allowed to post-process the query/answer vector $(q_1,a_1,...,q_i,a_i)$, the output predicate $h_{\vec{q_i}}$ is also generated in a $(\varepsilon,\delta)$--DP manner with respect to $\sigma_1,...,\sigma_L$, by \autoref{thm:post processing DP}. 

Given these $Q$ privately generated predicates, and since $L > \frac{1}{\varepsilon^2}\log\frac{2\varepsilon Q}{\delta}$ for large enough $Q$, by the generalization property of DP (\autoref{thm:dp-generalization}) we have that\footnote{Assuming $\delta \in (0,\varepsilon/4).$} with probability at least $1-\frac{\delta}{\varepsilon}=1-\Theta(\delta)$ it holds for any distribution $\mathcal{R}$ and for all $i \in [Q]$ that:
\begin{align}
    \label{eq:generalization-predicate}
    \left|\mathop{\mathbb{E}}_{\sigma\sim \mathcal{R}}\left[h_{\vec{q_i}}(\sigma)\right]-\frac{1}{L}\sum\limits_{j=1}^{L} h_{\vec{q_i}}(\sigma_j)\right| \leq 10\varepsilon = \frac{1}{10}
\end{align}

But if $\mathcal{R}$ is the uniform distribution, then $\mathbb{E}_{\sigma\sim \mathcal{R}}\left[h_{\vec{q_i}}(\sigma)\right]$ is simply the probability that $\mathcal{A}$ gives an accurate answer on the \textit{fixed} query sequence $\vec{q_i}$. Since $\mathcal{A}$ is an oblivious decider, \autoref{eq:generalization-predicate} implies that:
\begin{align}
\label{eq:expectation-predicate-bound}
\mathop{\mathbb{E}}_{\sigma\sim\mathcal{R}}\left[h_{\vec{q_i}}(\sigma)\right] \geq \frac{9}{10}
\end{align}
Further, $\frac{1}{L}\sum_{j=1}^{L} h_{\vec{q_i}}(\sigma_j)$ is the fraction of random strings that lead $\mathcal{A}_2$ to be correct. Thus, by \autoref{eq:expectation-predicate-bound}, this fraction is at least
$\left(\frac{9}{10}-\frac{1}{10}\right)L = 0.8L$ for all $i \in [Q]$. 
\end{proof}
We are now ready to prove the main theorem of this section. 
\begin{proof}[Proof of \autoref{thm:robust-decider-thm}]
Let us condition on the event that  \autoref{lem:majority_is_correct} holds, which happens with probability at least $1-\Theta(\delta)$. Then, for all $i \in [Q]$, the fraction of correct answers to query $q_i$ is either at least $0.8$, when $B_S(q_j,r) \neq \emptyset$, or at most $1-0.8 = 0.2$, otherwise. Now we need to account for the error introduced by subsampling, which is done via the following lemma:
\begin{lemma}
\label{lem:subsampling-concentration}
Let $L \in \mathbb{N}$ and suppose that at least an $\alpha$-fraction of $L$ structures are \emph{good}, where
$\alpha \in \{0.8,0.2\}$. Let $j_1,\dots,j_k$ be sampled independently and uniformly from $[L]$ (with replacement),
and define
\[
\bar N \;=\; \frac{1}{k}\sum_{t=1}^k X_t,
\qquad
X_t := \mathbf 1[\text{structure } j_t \text{ is good}].
\]
Then for $k \ge 40 \log(Q/\delta)$,
\[
\Pr\!\left[
\begin{array}{ll}
\bar N \le 0.6 & \text{if } \alpha = 0.8,\\[4pt]
\bar N \ge 0.4 & \text{if } \alpha = 0.2
\end{array}
\right]
\;\le\;
\delta/Q.
\]
\end{lemma}

\begin{proof}
The random variables $X_1,\dots,X_k$ are i.i.d.\ Bernoulli$(\alpha)$, hence
$\mathbb{E}[\bar N] = \alpha$.

\smallskip
\noindent
\textbf{Case 1: $\alpha = 0.8$.}
We apply a multiplicative Chernoff bound. For any $0<\gamma<1$,
\[
\Pr[\bar N \le (1-\gamma)\alpha]
\;\le\;
\exp\!\left(-\frac{\gamma^2 \alpha k}{2}\right).
\]
Setting $(1-\gamma)\alpha = 0.6$ gives $\gamma = 1/4$, and therefore
\[
\Pr[\bar N \le 0.6]
\;\le\;
\exp\!\left(-\frac{(1/4)^2 \cdot 0.8}{2}\, k\right)
\;=\;
\exp(-k/40).
\]

\smallskip
\noindent
\textbf{Case 2: $\alpha = 0.2$.}
By symmetry,
\[
\Pr[\bar N \ge 0.4]
\;\le\;
\exp(-k/40).
\]

\smallskip
\noindent
Choosing $k \ge 40 \log(Q/\delta)$ yields the claimed bounds.
\end{proof}

Finally, by \autoref{lem:bounded max of laplace distribution}, we require that the maximum Laplacian noise not exceed $0.1$ with high probability:
\begin{align}
    \label{eq:laplace-noise-bound}
    \Pr\left[|Z_i| > 0.1\right] = \Pr\left[|Z_i| > \frac{1}{k}\left(\ln (1) + 0.1k\right)\right] \leq e^{-0.1k}
\end{align}
Since our threshold for deciding is $\widehat{N}_i:=N_i+Z_i  \geq 0.5$, we can see that setting $k = \Omega(\log(Q/\delta))$ will make the probability in \autoref{eq:laplace-noise-bound} at most $\frac{\delta}{Q}$, implying, by union bound, that $\mathcal{A}_{\text{dec}}$ outputs the correct answer at every timestep $i \in [Q]$ with high probability.
\end{proof}

% \section{Proof of \autoref{thm:concentric-ann-robust-best}}
% \label{sec:proof-of-concentric-best}
% In this section we prove \autoref{thm:concentric-ann-robust-best} to analyze \autoref{alg:robust_ann_improved}.
\section{Improved Robust ANNS Algorithms with \texorpdfstring{$\forall$}{forall} guarantees}
\label{sec:for-all-main}
% \thnote{Done with this section. Techniques appear in \citep{cherapanamjeri2020adaptive, cherapanamjeri2024terminal} but our contributions are:
% % \begin{itemize}
% %     \item Simpler instantiation of framework.
% %     \item Showcase it can be used for search problems ubiquitously.
% %     \item Improve the runtime from $d(\log n + \log d)$ to $\log n + d\log d$.
% % \end{itemize}}
In this section, we will discuss another path to adversarial robustness for search problems --providing a \textit{for-all} guarantee. We will focus on the ANN problem for this section, due to its ubiquity and importance, as well as its amenity to the techniques we discuss.

% \textit{For-all} guarantees were popularized by the compressed sensing \citep{gilbert2007one} literature, in which
% one can construct a sketch that works for any vector, as long as the initial randomness is successfully selected,
% which happens with high probability. This is a case for us as well, because with high probability over the initial selection of hash functions,
% the data structure has good properties for all possible queries.
% The downside of this approach is that it requires creating a number of hash tables that is proportional
% to the dimension and may therefore not be suitable for very high dimensional spaces.
% We will first see how to do this for ANN in the Hamming Hypercube, and then present an algorithm for $\ell_p$ spaces that relies on discretization in Section \ref{sec:for-all-appendix} in the Appendix.

% This approach for the ANN problem was analyzed in \citep{cherapanamjeri2020adaptive}, where it was used to protect Monte Carlo algorithms against adaptively chosen queries for the problem of distance estimation. In addition, the framework was used in \citep{cherapanamjeri2024terminal} to solve the ANN problem in conjunction with maintaining a partition tree to calculate the terminal embeddings in sublinear time. Our application of these ideas simplifies their proof structure and offers small improvements to their runtime guarantees.

% \konote{This might have to be adjusted once we add sampling to obtain faster running time. Also Alex Andoni mentioned that this very similar
% to something from Jelani Nelson's paper, but I don't know which one.}
    
\subsection{A \textit{For-all} guarantee in the Hamming cube}
We present the Hamming Distance ANN case first because it is the most natural \textit{for-all} guarantee one can give. This is because the space we are operating over is discrete, and we can easily union-bound over all possible queries and only incur a cost polynomial to the dimension $d$ of the metric space.

\begin{theorem}
\label{thm:hamming-for-all-slow-query}
There exists an adversarially robust algorithm solving the $(c,r)$--ANN problem in the $d$--dimensional Hamming Hypercube that can answer every possible query correctly with probability at least $1-1/n^2$. The space requirements are $\widetilde{O}(d\cdot n^{1+\rho+o(1)})$, and the time required per query is $\widetilde{O}(d^2 \cdot n^\rho)$, where $\rho = 1/c$. 
\end{theorem}

\begin{proof}
First, let us recall the standard LSH in the Hamming Hypercube: We are given a point set $S \subseteq \{0,1\}^d$ with $|S| = n$. We receive queries $q \in \{0,1\}^d$. Our Locality Sensitive Hash family $\mathcal{H}$ is defined as follows: Pick some coordinate $i \in [d]$ and hash $x \in \{0,1\}^d$ according to $x_i \in \{0,1\}$. This function $h$ acts as a hyperplane separating the points in the hypercube into two equal halves, depending on the $i$-th coordinate. Sampling $h$ uniformly at random from $\mathcal{H}$ is equivalent to sampling $i \in [d]$ uniformly at random. We can easily see that $\mathcal{H}$ is an $(r,cr,p_1,p_2)$--LSH family, as:

\begin{align*}
\Pr_{h\sim\mathcal{H}}\left[h(p)=h(q)\right] = \frac{d-||p-q||}{d} = \begin{cases}
    \geq 1-\frac{r}{d}:=p_1,&\text{when }||p-q||\leq r\\
    \leq 1-\frac{cr}{d}:=p_2,&\text{when }||p-q||\geq cr
\end{cases}
\end{align*}

We now go through the typical amplification process for LSH families \citep{gionis1999similarity}. Instead of sampling just one coordinate, we sample $k$. And instead of sampling just one hash function, we sample $L$ different ones $h_1,...,h_L \in \mathcal{H}^k$ and require that a close point collides with $q$ at least once. With this scheme, we know that if we fix $q \in \{0,1\}^d$ and $p \in B_S(q,r)$ we have:
\begin{align*}
    \Pr\left[\exists i \in [L]\,:\, h_i(p) = h_i(q)\right] \geq 1-(1-p_1^k)^L
\end{align*}
Furthermore, if $||p-q|| \geq cr$, we must have:
\begin{align*}
    \Pr\left[\exists i\in [L]\,:\, h_i(q) = h_i(p)\right] \leq Lp_2^k
\end{align*}
Now, we want to guarantee that with high probability there doesn't exist any query $q \in \{0,1\}^d$ such that for all points $p \in B_S(q,r)$ we have $h_i(q) \neq h_i(p)$ for all $i \in [L]$. In other words, we want:
\begin{align*}
    \Pr\left[\exists q\in\{0,1\}^d\,:\, \forall p \in B_S(q,r)\,\forall i\in[L]:h_i(p) \neq h_i(q)\right] \leq \frac{1}{n}
\end{align*}
We can use the union bound to get:
\begin{align*}
&\Pr\left[\exists q\in\{0,1\}^d\,:\, \forall p \in B_S(q,r)\,\forall i\in[L]:h_i(p) \neq h_i(q)\right] \\
&\leq \sum\limits_{q \in \{0,1\}^d}\Pr\left[\forall p \in B_S(q,r)\,\forall i\in[L]:h_i(p) \neq h_i(q)\right]
\end{align*}
So it suffices to establish that for fixed $q \in \{0,1\}^d$ we have:
\begin{align*}
\Pr\left[\forall p \in B_S(q,r)\,\forall i\in[L]:h_i(p) \neq h_i(q)\right] \leq \frac{1}{n 2^d}
\end{align*}
We can weaken this statement and union-bound as follows:
\begin{align*}
\Pr\left[\forall p \in B_S(q,r)\,\forall i\in[L]:h_i(p) \neq h_i(q)\right] &\leq \Pr\left[\exists p \in B_S(q,r)\,\not\exists i\in[L]:h_i(p) = h_i(q)\right] \\
&\leq\sum\limits_{p \in B_S(q,r)}\Pr\left[\not\exists i\in[L]:h_i(p) = h_i(q)\right]\\
&\leq |B_S(q,r)|\cdot (1-p_1^k)^L \\
&\leq n(1-p_1^k)^L
\end{align*}
So it suffices to require that:
\begin{align}
\label{eq:L_requirement_hamming}
(1-p_1^k)^L \leq \frac{1}{n^2 2^d}
\end{align}
On the other hand, the expected number of points in $S \setminus B_S(q,cr)$ that we will see in the same buckets as $q$ is:
\begin{align}
    \EE\left[\left|p \in S\setminus B_S(q,cr)\mid \exists i\in[L]\,:\,h_i(p) = h_i(q)\right|\right] &= \sum\limits_{p \in S\setminus B_S(q,cr)} \Pr\left[\exists i \in [L]\mid h_i(p) = h_i(q)\right] \\
    &\leq nLp_2^k
    \label{eq:k_requirement_hamming}
\end{align}
We can now combine \autoref{eq:L_requirement_hamming} and \autoref{eq:k_requirement_hamming} to work out the values of $k$ and $L$. First, we want to get $O(L)$ time in expectation, so we require $p_2^k \leq 1/n$, which gives:
\begin{align*}
    k \geq \log_{1/p_2}(n)
\end{align*}
Now, let $p_1 = p_2^\rho$. Substituting, we resolve the value of $L$ as:
\begin{align*}
    L \geq n^\rho d \log n
\end{align*}
With that in place, we can see that our algorithm takes $O(L)$ time with high probability. Indeed, let $X$ be the number of points in $S \setminus B_S(q,cr)$ that are hashed to some common bucket with $q$. Using a simplified Chernoff bound, we have that:
\begin{align*}
\Pr\left[X \geq 10L\right] \leq 2^{-10L} = \frac{1}{n^{10dn^\rho}} \ll \frac{1}{n^{\Omega(1)}}
\end{align*}
which implies that our runtime per query is $O(L)$ with high probability. As for the value of the constant $\rho$ we have by definition that:
\begin{align*}
    \rho := \frac{\log p_1}{\log p_2} = \frac{\log \left(1-\frac{r}{d}\right)}{\log\left(1-\frac{cr}{d}\right)} \approx \frac{1}{c}
\end{align*}
Overall, evaluating our hash function requires $\widetilde{O}(\log n)$ time, and evaluating distances between points requires $O(d)$ time. We maintain $O(d\cdot n^\rho \log n)$ hash tables, meaning that on a single query we spend $O(d^2\cdot n^\rho \log n)$ time. For pre-processing, apart from storing the entire dataset in $dn$ space, we take $O(d\cdot n^{1+\rho +o(1)})$ space to construct our data structure. 
\end{proof}

\subsubsection{Improving the query runtime via sampling} We can improve the dependency on $d$ for the query runtime by using sampling to find a good bucket. The following theorem encapsulates this finding, reducing the runtime complexity by a factor of $d$:

\begin{theorem}
\label{thm:hamming-for-all-fast-query}
There exists an adversarially robust algorithm solving the $(c,r)$--ANN problem in the $d$--dimensional Hamming Hypercube that can answer all possible queries correctly with probability at least $1-1/n^2$. The space requirements are $\widetilde{O}(d\cdot n^{1+\rho+o(1)})$ and the time required per query is $\widetilde{O}(d \cdot n^\rho)$, where $\rho = 1/c$.
\end{theorem}

\begin{proof}
From our analysis above, we know that we take $L = n^\rho \cdot d\log n$ different hash functions. Consider some query $q$. We analyze the expected number of buckets that contain some point $p \in B_S(q,r)$. Let $X_q$ be a random variable representing the number of buckets $i \in [L]$ for which some point in $B_S(q,r)$ lies in bucket $i$. Define the following indicator random variable:
\begin{align*}
    \mathbbm{1}_i = \begin{cases}
        1,&\text{if some point $p \in B_S(q,r)$ lies in bucket $i \in [L]$}\\
        0,&\text{otherwise}
    \end{cases}
\end{align*}
By linearity of expectation, we can now write:
\begin{align*}
    \EE[X_q] &= \sum\limits_{i=1}^L \Pr[\mathbbm{1}_i = 1] \\
        &=\sum\limits_{i=1}^L \Pr\left[\bigcup\limits_{p\in B_S(q,r)}\{h_i(p) = h_i(q)\}\right]\\
        &\geq L\cdot p_1^k\\
        &= L\cdot (p_2)^{\rho k}\\
        &\geq \frac{L}{n^\rho} \\
        &= d\log n
\end{align*}

\noindent
By using the Chernoff bound, we can see that with high probability, $X_q$ is close to its expectation: 
\begin{align*}
    \Pr\left[X_q \leq \frac{1}{2}d\log n\right] \leq e^{-\frac{d\log n}{8}} = \frac{1}{n^{d/8}} \ll \frac{1}{n}
\end{align*}

Let us, then, condition on $X_q > \frac{1}{2}d\log n$. On query time, we can simply sample $m = \Theta(n^\rho \log n)$ buckets uniformly at random from $[L]$. We know that with probability at least $\frac{d\log n}{2n^\rho d\log n} = \frac{1}{2n^\rho}$, a single randomly selected bucket contains some point from $B_S(q,r)$. So, for all $m$ of the selections to not contain such a point, the probability is at most:
\begin{align*}
    \left(1-\frac{1}{n^\rho}\right)^{n^\rho \log n} \leq e^{-\log n} = \frac{1}{n}
\end{align*}
So, with probability at least $1-\frac{1}{n}$ we find a bucket containing a good point. Since, with high probability, the number of points in $P \setminus B_S(q,cr)$ in any bucket are $O(L)$, we see that this sampling method improves the query runtime to $O(n^\rho \log n)$. 
\end{proof}

\subsubsection{Utilizing the optimal LSH algorithm}
Our earlier exposition used the original LSH construction for the Hamming Hypercube \citep{indyk1998approximate} that achieves $\rho = 1/c$. We can also use the state-of-the-art approach from \citep{andoni2015optimal} that achieves $\rho = \frac{1}{2c - 1}$ in place of Theorem \ref{thm:hamming-for-all-slow-query}. This slightly improves the exponent on $n$:

\begin{theorem}
\label{thm:hamming-for-all-fast-query-optimal}
There exists an adversarially robust algorithm solving the $(c,r)$--ANN problem in the $d$--dimensional Hamming Hypercube that can answer all possible queries correctly with probability at least $0.99$. The space complexity is $O(d\cdot n^{1+\rho+o(1)})$, and the time required per query is $O(d \cdot n^\rho)$, where $\rho = \frac{1}{2c-1}$. These runtime guarantees hold with high probability.
\end{theorem}

The analysis is identical, so we will not repeat it again: Since the algorithm succeeds with constant probability, and we want it to succeed on all $2^d$ possible queries, we boost its success probability to $1-\frac{1}{100\cdot 2^d}$. This way, after the union bound, any query succeeds with probability at least $0.99$. Furthermore, the analysis of the sampling algorithm for improving the query runtime in \autoref{thm:hamming-for-all-fast-query} also remains the same. All that changes between using the standard Hamming norm LSH as opposed to the optimal one is the ratio $\rho := \frac{\log (1/p_1)}{\log (1/p_2)}$.

\subsection{Discretization of continuous spaces through metric coverings}
\label{sec:for-all-appendix}
The \textit{for-all} algorithm we presented as \autoref{thm:hamming-for-all-fast-query} cannot be applied outside of discrete spaces, however, because the key to our analysis was the union bound over all the possible queries. 

To simulate a similar argument for solving ANN in continuous, $\ell_p$ spaces, we can consider a strategy of discretizing the space. We place special ``marker'' points and guarantee that some version of the ANN problem is solvable around them. Then, when a query comes in, we find its corresponding marker point, and solve the ANN problem for it. We show that the answer we get is valid for the original query as well, so long as the ``neighborhood'' around the marker points is small enough. A similar strategy and covering construction appeared in \citep{cherapanamjeri2024terminal}, although they did not make algorithmic use of the ability to project any query point to the covering set. Instead, their algorithm deems it sufficient to be successful on every point on just the covering set.

\subsubsection{Metric coverings in continuous spaces}
\label{sec:lp_covering}
To initiate our investigation, we need the definition of a \textit{metric covering}:
\vspace{1mm}
\begin{definition}
Consider a metric space $\mathcal{M} = (\mathbb{R}^d, ||\cdot||_p)$ with metric $\mu$. Let $U \subset \mathbb{R}^d$ be a bounded subset. A set $\widehat{S} \subseteq \mathbb{R}^d$ is called an \textbf{$\Delta$-covering} of $U$ if for all $q \in U$ there exists some $\widehat{s} \in \widehat{S}$ such that
$$
||q-\widehat{s}||_p \leq \Delta
$$
\end{definition}

\noindent 
Suppose that $U$ is a bounded subset of $\mathbb{R}^d$. We can construct the following the following $\Delta$-covering of $U$: Let $C := \sup\limits_{x \in U}||x||_\infty$ and suppose $\{u_i\}_{i=1}^d$ is an orthonormal basis spanning $U$. We know that $||x||_\infty \leq C$ for all $x \in U$, so let us define:
\begin{align*}
    \widehat{S} &= \sum\limits_{i=1}^d \widehat{\alpha}_i u_i,\quad\text{where}\\
    &
    \widehat{\alpha}_i \in \{-C,-C+\varepsilon,...,C-\varepsilon,C\}
\end{align*}
for some choice of $\varepsilon$ that we will decide later. This is a standard construction for $\ell_2$ that we now extend to $\ell_p$ \citep{shalev2014understanding}. As defined, we have:
\begin{align*}
    \left|\widehat{S}\right| = \left(\frac{2C}{\varepsilon}\right)^d
\end{align*}

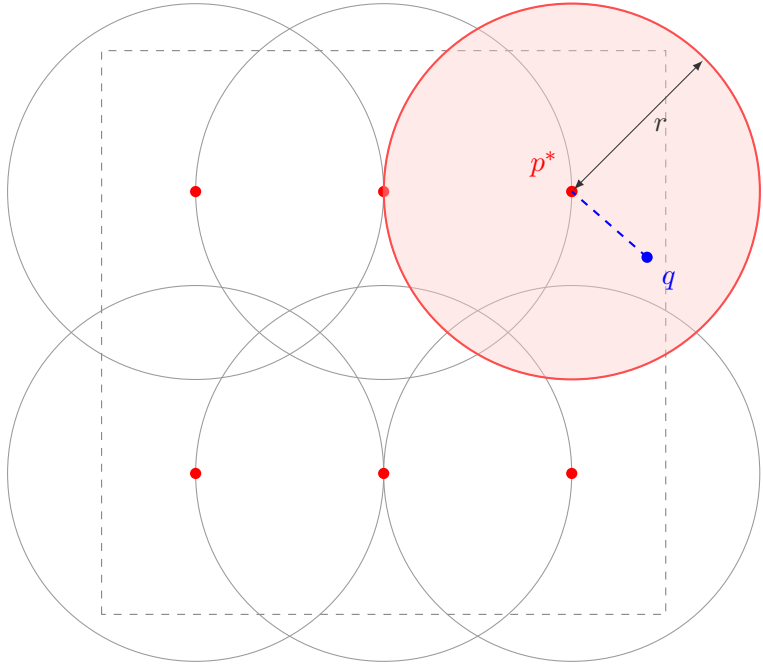
\begin{figure}[ht]
\centering
\begin{tikzpicture}[
    scale=1.25, % Scales the entire figure
    dot/.style={circle, fill, inner sep=1.5pt}
]
    % -- Define colors for easy modification --
    \colorlet{centercolor}{red}
    \colorlet{querycolor}{blue}
    \colorlet{covercolor}{black!40}      % Color for inactive covering circles
    \colorlet{activecovercolor}{red!70}  % Border for the active covering circle
    \colorlet{activefillcolor}{red!20}   % Fill for the active covering circle

    % -- Parameters for the grid --
    \def\bound{6} % Size of the bounding box
    \def\nx{3}    % Number of points in x
    \def\ny{2}    % Number of points in y
    
    % -- Calculated values --
    \pgfmathsetmacro{\dx}{\bound/\nx}
    \pgfmathsetmacro{\dy}{\bound/\ny}
    \pgfmathsetmacro{\radius}{\dx} % The radius r of the covering circles

    % -- 1. Draw background elements --
    \draw[gray, dashed] (0,0) rectangle (\bound, \bound);

    \foreach \i in {1,...,\nx} {
        \foreach \j in {1,...,\ny} {
            \pgfmathsetmacro{\x}{\i * \dx - \dx/2}
            \pgfmathsetmacro{\y}{\j * \dy - \dy/2}
            \draw[covercolor, thin] (\x, \y) circle (\radius);
        }
    }
    
    % -- 2. Draw the grid points --
    \foreach \i in {1,...,\nx} {
        \foreach \j in {1,...,\ny} {
            \pgfmathsetmacro{\x}{\i * \dx - \dx/2}
            \pgfmathsetmacro{\y}{\j * \dy - \dy/2}
            \node[dot, color=centercolor] at (\x, \y) {};
        }
    }

    % -- 3. Define and highlight the key elements --
    \coordinate (q) at (5.8, 3.8);     % The query point
    \coordinate (pstar) at (5, 4.5);   % The nearest center point, p*

    % Fill with a semi-transparent color
    \fill[activefillcolor, fill opacity=0.4] (pstar) circle (\radius);
    \draw[activecovercolor, thick] (pstar) circle (\radius);
    
    % Draw and label the key points with spaced-out labels
    \node[dot, color=querycolor, label={[querycolor]below right:$q$}] at (q) {};
    \node[dot, color=centercolor, label={[centercolor]above left:$p^*$}] at (pstar) {};
    
    % -- 4. Add final annotations --
    \draw[querycolor, dashed, thick] (q) -- (pstar);
    
    % Position the radius label to the side of the arrow
    \draw[<->, >=latex, black!80, shorten <=1pt, shorten >=1pt] (pstar) -- ++(45:\radius) node[midway, right=2pt] {$r$};

\end{tikzpicture}
\caption{An illustration of an $r$-covering.}
\label{fig:lp_covering_construction_final_scaled}
\end{figure}

Now, fix some $q \in U$. We can write:
\begin{align*}
q = \sum\limits_{i=1}^d \alpha_i u_i
\end{align*}
For all $i \in [d]$, let $\widehat{\alpha}_i$ be such that $\alpha_i \in \widehat{\alpha}_i \pm \varepsilon$. Let $\widehat{s} := \sum\limits_{i=1}^d \widehat{\alpha}_i u_i$. Now we have that:
\begin{align*}
||q-\widehat{s}||_p^p = \left|\left|\sum\limits_{i=1}^d (\alpha_i-\widehat{\alpha}_i)u_i\right|\right|_p^p = \sum\limits_{i=1}^d |\alpha_i - \widehat{\alpha}_i|^p \leq d\varepsilon^p
\end{align*}
Now, let us set:
\begin{align*}
    \varepsilon =\frac{\Delta}{d^{1/p}} \implies ||q-\widehat{s}||_p \leq \Delta
\end{align*}
Our construction thus has size:
\begin{align*}
    |\widehat{S}| = \left(\frac{2Cd^{1/p}}{\Delta}\right)^d
\end{align*}

\subsubsection{The robust ANN algorithm}
With this construction in mind, our algorithm for robust $(c,r)$--ANN in $\ell_p$ space follows as Algorithm \ref{alg:lp_discretization}. The algorithm remains agnostic to the specific LSH data structure that could be used to solve ANN in $\ell_p$ metric spaces obliviously \citep{charikar2002similarity, datar2004locality}, but assumes that the success probability over a set of queries in that data structure can be boosted by increasing the number of hash functions taken. This was the case for the Hamming norm as well.

\begin{algorithm}[ht]
\caption{Robust $\ell_p$ ANN through discretization}
\label{alg:lp_discretization}
\begin{algorithmic}[1]
\State \textit{Parameters: } Max-norm $C$, runtime/accuracy tradeoff $\Delta > 0$, LSH parameters $c,r > 0$.
\vspace{1mm}
\State Receive point dataset $S \subset U$ with $|S| = n$ from the adversary.
\State Let $\widehat{S}$ be a $\Delta$-covering of $U$ as constructed in Section \ref{sec:lp_covering}, and let $c'\gets \frac{cr-\Delta}{r+\Delta}$.
\State Initialize an LSH data structure $\mathcal{D}$ for solving $(c',r+\Delta)$--ANN that answers all queries in $\widehat{S}$ correctly with high probability.
\vspace{1mm}
\While{Adversary provides queries}
    \State Receive query $q \in U$ from the adversary.\
    \State Find $\widehat{s} \in \widehat{S}$ such that $||q-\widehat{s}||_p \leq \Delta$.
    \State Query $\mathcal{D}$ on $\widehat{s}$ and output whatever it outputs. 
\EndWhile

\end{algorithmic}
\end{algorithm}

\begin{theorem}
\label{thm:discretization-lp-C}
There exists an adversarially robust algorithm solving the $(c,r)$--ANN problem in the $(\mathbb{R}^d,\ell_p)$ metric space that can answer an unbounded number of adversarial queries. Assumming that the input dataset and the queries are all elements of $U = \{x \in \mathbb{R}^d \mid ||x||_p \leq C\}$ for some $C > 0$, the pre-processing  space is $\widetilde{O}(nT)$ and the time per query is $\widetilde{O}(T)$, where:
\begin{align}
T = O\left[d\cdot n^{\rho'}\log\left(\frac{Cd^{1/p}}{cr}\right)\right]
\end{align}
where:
\begin{align*}
\rho' = \frac{(10 + c)^2}{161c^2 - 20c - 100}
\end{align*}
\end{theorem}

\begin{proof}
First, to argue for correctness, let $q$ be any query. Suppose there exists some point $x \in S$ with $||x-q||_p \leq r$. Then, by triangle inequality it holds that:
\begin{align*}
    ||x-\widehat{s}||_p \leq ||x-q||_p + ||\widehat{s}-q||_p \leq \Delta + r
\end{align*}
Thus, with high probability, $\mathcal{D}$ will find some point $x' \in S$ with $||x'-\widehat{s}||_p \leq cr-\Delta$. For that point, we have that:
\begin{align*}
||x'-q||_p \leq ||x'-\widehat{s}||_p + ||\widehat{s}-q||_p \leq cr -\Delta + \Delta = cr
\end{align*}
Therefore, Algorithm \ref{alg:lp_discretization} will output a correct answer. If there doesn't exist such a point $x$, it is valid for our algorithm to output $\perp$, so are done. 

For the runtime, recall that $|\widehat{S}| \leq O(2Cd^{1/p}/\Delta)^d$. Hence, in order to guarantee success for all queries in $\widehat{S}$, a similar analysis as to the one for the Hamming Hypercube shows that $\mathcal{D}$ takes up:
$$
O\left[d\cdot n^{1+\frac{1}{2c'^2 - 1}}\log\left(\frac{2Cd^{1/p}}{\Delta}\right)\right]
$$
space for pre-processing and 
$$
O\left[n^{\frac{1}{2c'^2 - 1}}\log\left(\frac{2Cd^{1/p}}{\Delta}\right)\right]
$$
time per query processed, where
$$
c' := \frac{cr-\Delta}{r+\Delta}
$$
Note that we use the optimal LSH algorithm for $\ell_p$ spaces, which guarantees $\rho = \frac{1}{2c^2 - 1}$. Our only constraint is that we must have $\Delta < cr$. If we set $\Delta = \frac{c}{10}r$, we get a per-query runtime of:
\begin{align*}
O\left[n^{1+\frac{1}{2c'^2 - 1}}\log\left(\frac{20Cd^{1/p}}{cr}\right)\right],\quad\text{where}\, c'=\frac{9c}{10+c}
\end{align*}
\end{proof}

\subsubsection{Removing the dependency on the scale}
Our algorithm from Theorem \ref{thm:discretization-lp-C} crucially depends on $\log C$, where $C$ is a bounding box for the query and input point space in the $\ell_p$ norm. We can remove the dependency on $C$ by designing our covering to be data dependent, instead paying an additional logarithmic factor.

\begin{figure}[h]
    \centering
    \begin{tikzpicture}
        
        % Define coordinates of the three points, spaced apart
        \coordinate (A) at (0, 0);
        \coordinate (B) at (6, 1);
        \coordinate (C) at (3, 5);
        
        % Function to draw a 3x3 grid centered on a given point
        \newcommand{\drawCenteredGrid}[1]{
            % Draw 3x3 grid of 1x1 squares centered around the point
            \foreach \x in {-1.5, -0.5, 0.5} {
                \foreach \y in {-1.5, -0.5, 0.5} {
                    \draw[red] ([shift={(\x,\y)}]#1) rectangle ++(1,1);
                }
            }
            % Bold center point
            \fill[black] (#1) circle (3pt);
        }
        
        % Draw grids around each of the three points
        \drawCenteredGrid{A};
        \drawCenteredGrid{B};
        \drawCenteredGrid{C};
        
        % Optional: label the points
        \node[below left] at (A) {A};
        \node[below left] at (B) {B};
        \node[above right] at (C) {C};
        
    \end{tikzpicture}
    \caption{Data-Dependent Discretization of the input query space.}
    \label{fig:enter-label}
\end{figure}
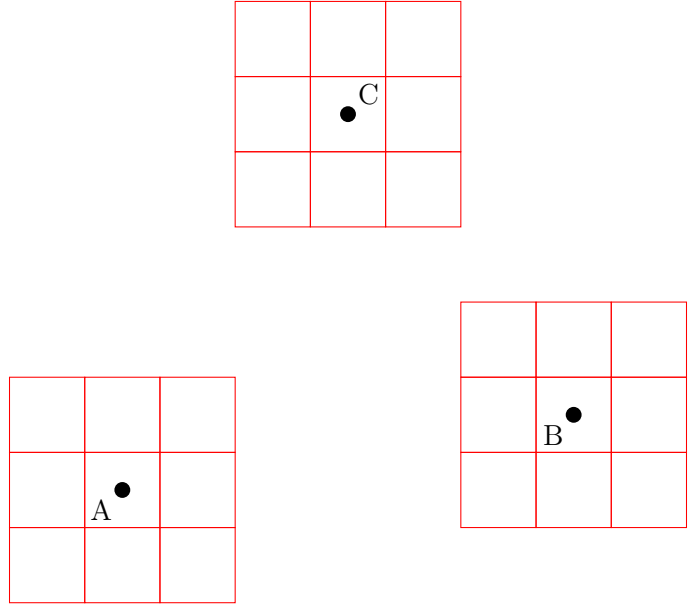

Our new covering $\widehat{S'}$ will be a collection of $n$ $\Delta$-coverings, as constructed in Algorithm \ref{alg:lp_discretization}, each one discretizing the $r$-ball around a point $p \in S$. The number of points in this new covering is:
\begin{align}
    |\widehat{S'}| \leq O\left[n\cdot\left(\frac{r\cdot d^{1/p}}{cr}\right)^d\right] = O\left[n\cdot\left(\frac{d^{1/p}}{c}\right)^d\right]
\end{align}

Note that the size of this covering improves upon the $(nd)^d$ size of the covering given in \citep{cherapanamjeri2024terminal}, which results in a slightly better runtime. This new covering notably does not cover every possible query. However, it covers exactly the queries we care about. This improved covering leads to the following \textit{for-all} guarantee for robust ANN:

\begin{theorem}
\label{thm:discretization-lp-no-C}
There exists an adversarially robust algorithm solving the $(c,r)$--ANN problem in the $(\mathbb{R}^d,\ell_p)$ metric space that can answer an unbounded number of adversarial queries. The pre-processing time / space is $\widetilde{O}(nT)$ and the time per query is $\widetilde{O}(T/d)$, where:
\begin{align}
T = O\left[d\cdot n^{\rho'}\left(d\log d + \log n\right)\right]
\end{align}
where:
\begin{align*}
\rho' = \frac{1}{2c'^2 - 1} = \frac{(10 + c)^2}{161c^2 - 20c - 100}
\end{align*}
\end{theorem}

\begin{proof}
We distinguish between two cases:

\begin{enumerate}
    \item If a query $q$ is not included in any $B_S(p,r)$ for any $p \in S$, then the answer can safely be $\perp$ because $B_S(q,r) = \emptyset$ necessarily. Thus, we can just run the default LSH algorithm and simply output whatever it outputs.
    \item Otherwise, a query $q$ can be included in some $B_S(p,r)$ for some $p \in S$. Then, suppose $\widehat{s'} \in \widehat{S'}$ is a point in our covering such that $||q-\widehat{s'}||_p \leq \Delta$. Then:
    \begin{align}
        ||p-\widehat{s'}||_p \leq ||p-q||_p + ||\widehat{s'}-q||_p \leq r + \Delta
    \end{align}
    Thus, as we argued before, with high probability $\cD$ finds some point $x \in S$ with $||x - \widehat{s'}||_p \leq cr - \Delta$, and for that point we have:
    \begin{align}
    ||x-q||_p \leq ||x-\widehat{s'}||_p + ||\widehat{s'} - q||_p \leq cr-\Delta + \Delta = cr
    \end{align}
    which means our algorithm will output a correct answer. 
\end{enumerate}
As before, our algorithm's space and runtime guarantees scale with $\log |\widehat{S'}|$.
\end{proof}

\end{document}